\documentclass{article}
\usepackage[utf8]{inputenc}
\RequirePackage{xcolor}
\usepackage{authblk}
\usepackage[sort]{natbib}
\usepackage{algorithm}
\usepackage{algorithmicx}
\usepackage{algpseudocode}
\usepackage{url}
\usepackage{bm}
\usepackage{graphicx}
\usepackage{amsmath}
\usepackage{amsthm}
\usepackage{amssymb,stackengine,graphicx}
\usepackage{thmtools,thm-restate}
\usepackage[colorlinks]{hyperref}
\usepackage[a4paper,left=3cm,right=3cm,top=3cm,bottom=3cm]{geometry}
\hypersetup{
    colorlinks,
    linkcolor={red!50!black},
    citecolor={blue!80!black},
    urlcolor={blue!80!black}
}

\newtheorem{mypro}{Proposition}
\newtheorem{mythm}{Theorem}
\newtheorem{mylem}{Lemma}
\newtheorem{corollary}{Corollary}
\theoremstyle{definition}
\newtheorem{mydef}{Definition}
\theoremstyle{remark}
\newtheorem{remark}{Remark}

\author[1]{Yukun Song}
\author[2]{Carey E. Priebe}
\author[1]{Minh Tang}
\affil[1]{Department of Statistics, North Carolina State University}
\affil[2]{Department of Applied Mathematics and Statistics, Johns
  Hopkins University}
\date{}
\title{Independence testing for inhomogeneous random graphs}

\begin{document}
\maketitle

\begin{abstract}
  Testing for independence between graphs is a problem that arises naturally in social network analysis and neuroscience. In this paper, we address independence testing for inhomogeneous Erd\H{o}s-R\'{e}nyi random graphs on the same vertex set. We first formulate a notion of pairwise correlations between the edges of these graphs and derive a necessary condition for their detectability. We next show that the problem can exhibit a statistical vs. computational tradeoff, i.e., there are regimes for which the correlations are statistically detectable but may require algorithms whose running time is exponential in $n$, the number of vertices. Finally, we consider a special case of correlation testing when the graphs are sampled from a latent space model (graphon) and propose an asymptotically valid and consistent test procedure that also runs in time polynomial in $n$. 
\end{abstract}

\section{Introduction}
Detecting correlation and testing for independence between
Euclidean vectors is one of the most classical and widely studied inference problems in multivariate statistics. See \citet[Chapter 11]{anderson} and \citet{kendall} for standard references when the data are low-dimensional and \citet{leung_drton,gretton_hsic,correlation_distances} for some examples of recent results in the high-dimensional setting. In contrast to the above mentioned literature, independence testing for graphs is much less studied. Part of the difficulty lies in choosing an appropriate notion of correlation for graph-valued data. In this paper, we are motivated by the problem of, given two graphs on the same vertex set, determining whether or not the presence of an edge between any pair of vertices in the first graph is {\em stochastically} independent of the presence of the corresponding edge in the second graph. Two graphs are thus said to be independent if all edges in the first are {\em pairwise} independent of the corresponding edges in the second, and are said to be correlated otherwise.

The above notion of pairwise edge correlation, while rather simple, does appear in many real data applications. For example, two networks on different social platforms are likely to share a common subset of users whose induced sub-networks are correlated in that two users are more likely to be linked in one social platform if they are already linked in the other social platform. As another example, the graphs constructed from the wiring diagrams of the mushroom body for the left and right hemispheres of the {\em Drosophila} larval are highly correlated as a pair of neurons are more likely to send signals to each other in the left hemisphere if their correspondences also send signals to each other in the right hemisphere \citep{eichler2017complete,winding_science}. Finally, two knowledge graphs constructed using documents on the same topics but written in different languages are also highly correlated \citep{priebe2009fusion,haghighi}.

In this paper, we consider independence testing in the setting where the graphs are, {\em marginally}, inhomogeneous Erd\H{o}s-R\'{e}nyi random graphs while the pairwise edges correlations are described by the entries of a symmetric matrix $R$ (see Definition~\ref{def_1}). Two graphs are then independent if $R_{ij} \equiv 0$ for all $\{i,j\}$ and are correlated if $|R_{ij}| > 0$ for {\em some}  $\{i,j\}$. Equivalently, we are interested in testing the null hypothesis $\mathbb{H}_0 \colon \|R\|_{F} = 0$ against the alternative hypothesis $\mathbb{H}_{A} \colon \|R\|_{F} > 0$.  

Our contributions are as follows. We first show that there exists a procedure for deciding between $\mathbb{H}_0$ and $\mathbb{H}_A$ with asymptotically {\em vanishing} type-I and type-II errors only if $\|R\|_{F} \rightarrow \infty$ as $n \rightarrow \infty$ (see Section \ref{sec:statistical_limit}). We also describe two related examples where the condition $\|R\|_{F} \rightarrow \infty$ is sufficient for one example but not sufficient for the other. We next show that the problem exhibits a statistical vs. computational tradeoff, i.e., there are regimes for which $\|R\| \rightarrow \infty$ that are statistically detectable but may require running time which scales exponentially with $n$. We achieve this through a polynomial time reduction of the planted clique problem, which is well-known to be computationally hard \citep{alon1998finding,barak2019nearly}, to our correlation testing problem (see Theorem~\ref{thm:gap}). Finally, we consider a special case of our correlation testing problem in which the graphs are sampled from the graphon or latent space model \citep{hoff2002latent,lovasz12:_large,bollobas2007phase} and propose an asymptotically valid and consistent test procedure that also runs in time polynomial in $n$ (see Section~\ref{graphon Model}).

We evaluate the performance of the proposed test procedures through simulations and show that they exhibits power even for moderate values of $n$ and small values of $n^{-1}\|R\|_{F}$. We then apply these procedures to two real data experiments. In the first experiment, we analyze the electrical and chemical connectomes for the {\em C. elegans} worm and show that there are significant correlations between the two connectomes; this confirms the observation made in \cite{chen2016joint} wherein the authors showed that, for their vertex nomination tasks, using both connectomes lead to better accuracy. In the second experiment, we analyze two Wikipedia
hyperlink graphs constructed using documents on the same topics but in different languages and show that, by estimating the correlations between the graphs, we obtain more accurate links prediction.

\subsection{Related Works}
Existing research on independence testing for a pair of graphs is quite limited. Among the current literature, only \citet{xiong2019graph} considered independence testing where the notion of correlation is similar to that described here, but the setting in \citet{xiong2019graph} is much more restrictive as they assumed that the pairwise correlations are the same for all edges, i.e., $R_{ij} \equiv c$ for some constant $c$ and furthermore each observed graph is, {\em marginally}, distributed according to a stochastic blockmodel \citep{holland1983stochastic}. Their results will be a special case of Theorem~\ref{thm:SBM} in this paper. 

Detecting pairwise edges correlation between two graphs is also a central focus in many graph matching algorithms. More specifically, given a pair of {\em unlabeled} adjacency matrices $A$ and $B$, graph matching aims to find a mapping between their vertex sets that best preserves their common structures, for example by minimizing the Frobenius norm error $\|\Pi A \Pi^{\top} - B\|_{F}$ over all {\em permutation} matrices $\Pi$. Recent results show that many graph matching algorithms have {\em average} running times that are polynomial in $n$, the number of vertices, whenever the graphs are sufficiently correlated; see e.g., \citep{wu2020testing,lyzinski2019matchability,lyzinski2016graph,pedarsani2011privacy,onaran2016optimal,korula2014efficient} and the references therein. In contrast, the worst case running time for these algorithms could be exponential in $n$ if the graphs are independent. Because the correspondence between the vertex sets is assumed {\em unknown} in graph matching but is assumed known in the context of our paper, our technical challenges and results are quite different from those for graph matching. Indeed, a common assumption in graph matching (see e.g., \cite{wu2020testing,lyzinski2016graph,pedarsani2011privacy}) is that $R_{ij} \equiv c$ for some constant $c$, and furthermore that the observed graphs are, {\em marginally}, Erd\H{o}s-R\'{e}nyi with the same edge probability $p$.

\section{Background and Setting}
\label{sec:background}
We now formally introduce the hypothesis testing problem considered in this paper. We begin by describing the notion of edge correlated inhomogeneous Erd\H{o}s-R\'{e}nyi graphs. Note that all graphs considered in this paper are simple, undirected graphs, i.e., there are no self-loops and no multiple edges.
\begin{mydef}\label{def_1}
  Let $n \in \mathbb{N}$. Let $P\in [0,1]^{n\times n}$ and
  $R\in[{-1},1]^{n\times n}$ be {\em symmetric} matrices where $R$
  satisfies the constraint
  \begin{equation}
      \label{eq:r_ij_bound1}
        R_{ij}\geq-\min\Big\{\frac{1-P_{ij}}{P_{ij}},\frac{P_{ij}}{1-P_{ij}}\Bigr\}, \quad \text{for all $i,j$}. 
  \end{equation}
  We say that $(A,B)$ are $R$-correlated heterogeneous Erd\H{o}s-R\'{e}nyi graphs on $n$ vertices with probability matrix $P$ and correlation matrix $R$, denoted by $(A,B)\sim R$-$\mathrm{ER}(P)$, if 
    \begin{enumerate}
        \item
            $A$ is the adjacency matrix for an inhomogeneous Erd\H{o}s-R\'{e}nyi graph on $n$
            vertices with probability matrix $P$, that is, $A$ is a $n
            \times n$ symmetric binary matrix whose (upper triangular) entries are
            independent Bernoulli random variables with success
            probabilities $\{P_{ij}\}_{i < j}$. 
        \item
            $B$ is the adjacency matrix for another inhomogeneous Erd\H{o}s-R\'{e}nyi graph on $n$ vertices with probability matrix $P$.
        \item The pairs $\{(A_{ij},B_{ij})\}_{i < j}$ are {\em mutually
            independent} bivariate random variables. Here $A_{ij}$ and
          $B_{ij}$ are the $ij$th entry of $A$ and
          $B$, respectively. 
        \item For any $i < j$, $A_{ij}$ and $B_{ij}$ are correlated with
          Pearson correlation $R_{ij}$.
    \end{enumerate}
\end{mydef}
\begin{remark}\label{remark_ext}
Definition \ref{def_1} assumes that $A$ and $B$ have the same marginal distribution. In Section \ref{different} we will consider a slight extension of this model where we allow for $A$ and $B$ to have different marginal distributions. More specifically, we say that $(A, B)$ are generated from the $R$-ER$(P,Q)$ model if conditions 1 and 2 in Definition \ref{def_1} are replaced by the assumption that $A$ is marginally $P$ and $B$ is marginally $Q$, respectively. The remaining conditions 3 and 4 remain unchanged. Let $\phi(x) = \sqrt{x/(1-x)}$ for $x \in (0,1)$. We then note that for the correlation $R_{ij}$ to be valid in the $R$-ER$(P,Q)$ model, we assume
  \begin{equation}
  \begin{split}
      \label{eq:r_ij_p_q_bound}
  -\min\Bigl\{\phi(P_{ij})\phi(Q_{ij}),\frac{1}{\phi(P_{ij})\phi(Q_{ij})}\Bigr\}
  \leq R_{ij}
  \leq\min\Bigl\{\frac{\phi(P_{ij})}{\phi(Q_{ij})},\frac{\phi(Q_{ij})}{\phi(P_{ij})}\Bigr\}
  \end{split}
  \end{equation}
as a replacement for the condition $R_{ij}\geq-\min\{\tfrac{1-P_{ij}}{P_{ij}},\tfrac{P_{ij}}{1-P_{ij}}\}$ in Definition \ref{def_1}.
\end{remark}

Correlated inhomogeneous ER graphs are widely studied in the graph
matching literature, see e.g., \citet{fishkind2019alignment,sussman2018matched,lyzinski2019matchability} and the references therein. However, as we allude to in the introduction, for graph matching the correlation matrix $R$ is usually assumed to be a constant matrix, i.e., $R_{ij} \equiv c$ for all $\{i,j\}$. 

Let $\mathbb{A}_n$ denote the set of $n \times n$ hollow symmetric binary matrices. Given a pair of $R$-ER$(P)$ graphs and
symmetric matrices $[a_{ij}]$, $[b_{ij}] \in \mathbb{A}_n$, the joint likelihood for the adjacency matrices $A=[A_{ij}]$ and $B=[B_{ij}]$ is 
    \begin{gather*}
        \mathbb{P}(A=[a_{ij}],B=[b_{ij}])
        =
        \prod_{i < j}
        \mathbb{P}(A_{ij}=a_{ij},B_{ij}=b_{ij})
    \end{gather*}
    where $\mathbb{P}(A_{ij} = a, B_{ij} = b)$ for $1 \leq i < j \leq
    n$ is given by
    \begin{equation}\label{distribution P_n}
        \mathbb{P}(A_{ij}=a,B_{ij}=b)=
        \begin{cases}
            P_{ij}^2 + P_{ij}(1-P_{ij})R_{ij},          &a = b= 1\\
            (1-P_{ij})^2 + P_{ij}(1-P_{ij})R_{ij},      &a = b=0, \\
            P_{ij}(1-P_{ij})(1-R_{ij}),  & a \not = b.
        \end{cases}
    \end{equation}

Let $(A,B)$ be a pair of $R$-ER$(P)$ graphs with unknown correlation matrix $R$ and edge probabilities matrix $P$. We want to test the  hypotheses $\mathbb{H}_0 \colon R = 0$ against $\mathbb{H}_A \colon R \not = 0$. This is equivalent to testing
 \begin{equation}
   \label{eq:test_formulation}
        \mathbb{H}_0:\|R\|_{\ast} =0 \quad \text{against} \quad \mathbb{H}_A:\|R\|_{\ast}>0.
 \end{equation}
for any choice of matrix norm $\|\cdot\|_{\ast}$. We will use the Frobenius norm as it is one of the simplest and most widely used norms while also yielding useful and interesting theoretical results in the setting of the current paper; see in particular Theorem~\ref{thm:main1} below. We expect that other norms, such as the spectral or infinity norms, will lead to results that are related but distinct from those presented here and we leave this for future work.
 
Let $T$ be any test procedure for the hypothesis in Eq.~\eqref{eq:test_formulation}. A desirable property for $T$ is consistency, that is, as the number of vertices $n$ approaches infinity, we want $T$ to have {\em both} vanishing Type-I error and vanishing Type-II error. More specifically, along a sequence of edge probabilities matrices and correlation matrices $(P_n,R_n)$, a test procedure is consistent if its error rate converges to $0$ as $n\to\infty$, that is
    \begin{equation*}
        \lim_{n\to\infty} \Bigl(1 - \mathbb{P}(\text{reject} \,\,
        \mathbb{H}_0 \mid \|R_n\|_{F} > 0) +
        \mathbb{P}(\text{reject}\,\, \mathbb{H}_0 
        \mid \|R_n\|_{F} = 0)\Bigr)=0.
    \end{equation*}

\section{Statistically Limit}
\label{sec:statistical_limit}
\subsection{Detectability Threshold}
\label{sec:detectability}
We first derive a necessary condition for the hypothesis testing problem in Eq.\eqref{eq:test_formulation} to be statistically detectable. Our approach is based on the second moment method for the ratio of the likelihood when $\|R\|_{F} = 0$ against the likelihood when $\|R\|_{F} > 0$. See \citet{wu2018statistical} for an elegant survey of the second moment method and its application in deriving detectability thresholds for statistical problems with planted structures. 
    
The second moment method is motivated by the notion of contiguity between distributions. Contiguity is an asymptotic generalization of absolutely continuity characterized as follows.
    \begin{mydef} 
    Let $\{\mathcal{P}_n\}_{n \geq 1}$ and $\{\mathcal{Q}_n\}_{n \geq 1}$ be two sequence of probability measures. We say that $\{\mathcal{P}_n\}$ is contiguous with respect to $\{\mathcal{Q}_n\}$ if $\mathcal{Q}_n(S_n)\rightarrow 0$ implies $\mathcal{P}_n(S_n)\rightarrow 0$ for every sequence of measurable sets $S_n$.
    \end{mydef}
If the probability distribution under the alternative hypothesis is contiguous with respect to the distribution under the null hypothesis, then there does not exist a valid and consistent test procedure. Indeed, due to contiguity, any rejection region with vanishing type I error under the null hypothesis will necessarily have vanishing power under the alternative hypothesis. For more on the definition of contiguity and its properties, see \citet{van2000asymptotic}.

Let $\{\mathcal{P}_n\}_{n \geq 1}$ and $\{\mathcal{Q}_n\}_{n \geq 1}$ be the sequences of {\em joint distributions} for pairs of adjacency matrices $(A_n,B_n)_{n \geq 1}$ from the $R$-ER($P$) random graphs model with $\|R\|_{F} > 0$ (for $\mathcal{P}_n$) and $\|R\|_{F} = 0$ (for $\mathcal{Q}_n$), respectively. We emphasize that, for simplicity of notations, we have dropped the index $n$ from the matrices $R$ and $P$.

It is well known that $\{\mathcal{P}_n\}$ is contiguous with respect to $\{\mathcal{Q}_n\}$ (see e.g., Eq.(13.3) of \cite{wu2018statistical}) if the second moment of the likelihood ratio between $\mathcal{P}_n$ and $\mathcal{Q}_n$ is bounded i.e., that $$ \limsup_{n \rightarrow \infty} \mathbb{E}_{(A_n,B_n) \sim\mathcal{Q}_n}\Bigl[\Bigl(\frac{\mathcal{P}_n(A_n,B_n)}{\mathcal{Q}_n(A_n,B_n)}\Bigr)^2\Bigr] < \infty.$$

We then have the following result (see the appendix for a proof). 
\begin{mythm}
  \label{thm:main1}
  Let $\{\mathcal{P}_n\}_{n \geq 1}$ be the sequence of {\em joint distributions} for pairs of adjacency matrices $(A_n,B_n)_{n \geq 1 }$ from the $R$-$\mathrm{ER}(P)$ random graphs model with $\|R\|_{F} > 0$. Similarly, let $\{\mathcal{Q}_n\}$ be the joint distribution of the $(A_n,B_{n})_{n \geq 1}$ when $\|R\|_{F} = 0$. Then
    \begin{equation*}
        \begin{split}
            \mathbb{E}_{(A_n,B_n) \sim\mathcal{Q}_n}\Bigl[\Bigl(\frac{\mathcal{P}_n(A_n,B_n)}{\mathcal{Q}_n(A_n,B_n)}\Bigr)^2\Bigr]&=\prod_{i<j}(1+R_{ij}^2).
        \end{split}
  \end{equation*}
Therefore if $\limsup_{n \rightarrow \infty} \|R\|_{F} < \infty$ then $\mathcal{P}_n$ is contiguous with respect to $\mathcal{Q}_n$.
\end{mythm}

Theorem~\ref{thm:main1} showed that the condition $\limsup \|R\|_{F} = \infty$ as $n \rightarrow \infty$ is necessary for the existence of a
consistent test procedure for testing Eq.~\eqref{eq:test_formulation}. We now present an example of a hypothesis testing problem for which $\limsup
\|R\|_{F} = \infty$ is also sufficient. Consider a $R$-ER$(P)$ model with $P_{ij} \equiv p$, i.e., the marginal distribution for the observed graphs is
Erd\H{o}s-R\'{e}nyi. Suppose also that $p \geq c_0$ for some constant $c_0 > 0$ not depending on $n$. Next, suppose that $n$ is even and let $\sigma = (\sigma_1, \dots, \sigma_n)$ be such that $\sigma_{i} \in \{-1,1\}$ for all $i$ and $\sum_{i=1}^{n} \sigma_i = 0$. Given $\sigma$, the correlation matrix $R$ has entries
    \begin{equation}
      \label{eq:rn_example1}
      R_{ij} = 
        \begin{cases}
            r,  & \text{if $\sigma_i=\sigma_j$} \\
            0,  & \text{otherwise}.
        \end{cases}
    \end{equation} 
    We are interested in testing the hypothesis
    \begin{align*}
        \mathbb{H}_0 \colon r=0 \quad \text{against} \quad
        \mathbb{H}_A \colon r\neq0 
    \end{align*}
      and this is a special case of Eq.~\eqref{eq:test_formulation} where $$ \frac{(n^2-2n)}{4}r^2 \leq \|R\|_F^2 = \sum_{i\neq j} r^2 1_{\{\sigma_i=\sigma_j\}} \leq n^2r^2.$$
    Now consider the matrix $S = A \circ B$ with elements $S_{ij} = A_{ij}B_{ij}$. The $S_{ij}$ are then {\em independent} Bernoulli random variables with
    \begin{equation*}
        S_{ij} \sim 
        \begin{cases}
            \text{Bernoulli}(p^2 + rp(1-p)),    & \sigma_i=\sigma_j \\
            \text{Bernoulli}(p^2),    & \sigma_i\neq\sigma_j.
        \end{cases}
    \end{equation*}
In other words $S$ is the adjacency matrix of an Erd\H{o}s-R\'{e}nyi graph with edge probabilities $p^2$ under $\mathbb{H}_0$ and is the adjacency matrix of a $2$-blocks stochastic blockmodel graph under $\mathbb{H}_A$. Let $\lambda_1(S)$ be the largest eigenvalue of $S$. Then by \cite{furedi1981eigenvalues} and \cite{tang2018eigenvalues},  we have
    \begin{gather}
        \lambda_1(S) - np^2 \longrightarrow N(1 - p^2, 2p^2(1 - p^2))
        \quad \text{under $\mathbb{H}_0$}, \\
        \lambda_1(S) - np^2 - \frac{1}{2}nrp(1-p) \longrightarrow
        N(\eta, \gamma) \quad \text{under $\mathbb{H}_A$}.
    \end{gather}
Here $\eta$ and $\gamma$ are bounded constants. Therefore $\lambda_1(S)$ yields a consistent test procedure for testing $\mathbb{H}_0 \colon r = 0$ against $\mathbb{H}_A \colon r \not = 0$ whenever $\|R\|_{F} \asymp |nr| \rightarrow \infty$. More specifically, as both $A$ and $B$
are {\em marginally} Erd\H{o}s-R\'{e}nyi graphs, first estimate $p$ via $$\hat{p} = \frac{1}{n(n-1)}\sum_{i < j} (a_{ij} + b_{ij}).$$
Next, define $T(A, B) = |\lambda_1(S) - n \hat{p}^2|$. Then $T(A, B)$ is bounded in probability under $\mathbb{H}_0$ and $T(A, B)$ diverges under $\mathbb{H}_A$ if $\|R\|_{F} \rightarrow \infty$. We summarized the above discussion in the following result.
    
\begin{mypro}
\label{pro:1}
  Let $(A, B)$ be a pair of $R$-correlated Erd\H{o}s-R\'{e}nyi graphs with marginal edges probability $p$ where $p > 0$ does not depend on $n$ and suppose that $R$ is of the form in Eq.~\eqref{eq:rn_example1}. Then there exists a consistent test procedure for testing $\mathbb{H}_0 \colon r = 0$ against $\mathbb{H}_A \colon r \not = 0$ if and only if $\|R\|_{F} \rightarrow \infty$ as $n \rightarrow \infty$. 
\end{mypro}

\begin{remark}
  \label{rem:not_sufficient}
We end this subsection with an example of a correlation matrix $R$ for which the condition $\|R\|_{F} \rightarrow \infty$ is {\em not} sufficient to guarantee the existence of a consistent test procedure. Fix a constant $\epsilon \geq 0$. Let $(A, B)$ be $R$-correlated Erd\H{o}s-R\'{e}nyi graphs with marginal edges probability $p = 0.5$ where $R$ is a symmetric matrix whose (upper triangular) entries are  iid random variables with $\mathbb{P}[R_{ij} = \epsilon] = \mathbb{P}[R_{ij} = - \epsilon] = 0.5$. 

Given $(A, B)$, detecting correlation between $A$ and $B$ is equivalent to testing $\mathbb{H}_0 \colon \epsilon = 0$ against $\mathbb{H}_A \colon \epsilon > 0$. However, as $R$ is unknown and we only observed $(A, B)$, the probabilities of a given pair $(A,B)$ under either $\mathbb{H}_0$ or $\mathbb{H}_A$ are the same. More specifically, under $\mathbb{H}_0$ we have
    \begin{equation}
        \mathbb{P}_{\mathbb{H}_0}(A_{ij}=a,B_{ij}=b) = \frac{1}{4},
        \quad \text{for all $(a,b) \in \{(0,0),(0,1),(1,0),(1,1)\}$}. 
    \end{equation}
 Meanwhile under $\mathbb{H}_A$, if we first conditioned on $R_{ij}$ then
    \begin{equation}
        \mathbb{P}(A_{ij}=a,B_{ij}=b|R_{ij}) = 
            \begin{cases}
                \frac{1 + R_{ij} }{4}, &\text{if $(a,b) \in \{(0,0), (1,1)\}$}\\
                \frac{1 - R_{ij}}{4}, &\text{if $(a,b) \in \{(1,0),(0,1)\}$}
           \end{cases}
    \end{equation}
    and hence, as $R_{ij}$ is unobserved, 
    \begin{equation*}
      \begin{split}
        \mathbb{P}_{\mathbb{H}_A}(A_{ij} = a, B_{ij} = b) &=
        \frac{1}{2} \mathbb{P}(A_{ij} = a, B_{ij} = b \mid R_{ij} = \epsilon) +
        \frac{1}{2} \mathbb{P}(A_{ij} = a, B_{ij} = b \mid R_{ij} = -\epsilon)
        \\ &= \frac{1}{4}.
        \end{split}
    \end{equation*}

There is thus no consistent test procedure for $\mathbb{H}_0 \colon \epsilon = 0$ against $\mathbb{H}_A \colon \epsilon \not = 0$ under this assumed correlation structure for $R$ even though $\|R\|_{F} = n \epsilon$ for {\em any} realization of $R$ (so that $\|R\|_{F} \rightarrow \infty$ as $n\rightarrow \infty$ for any fixed $\epsilon > 0$). The condition $\|R\|_{F} \rightarrow \infty$ is therefore not sufficient.
\end{remark}

\subsection{Computational Feasibility}
\label{sec:computational}
The results in Section~\ref{sec:detectability} provides a necessary condition for statistical detectability, i.e., the existence of a consistent test procedure for testing $\mathbb{H}_0 \colon \|R\|_{F}$ against $\mathbb{H}_{A} \colon \|R\|_{F} > 0$. However, there are numerous problems that are statistically feasible with parameter regimes for which there are no known {\em computationally efficient} procedures for solving them. Examples include community detection \citep{banks2016information}, sparse PCA \citep{sparse_pca} and estimation in spiked tensor models \citep{statistical_limit_tensor}. We now present an example of this phenomenon in the context of independence testing. In particular, we show the presence of a statistical vs. computational gap by transforming the independence testing problem to the following well-known planted clique problem in theoretical computer science.

Let $A$ be an Erd\H{o}s-R\'{e}nyi graph on $n$ vertices with common edges probability $p$. Next, given an integer $s_0 \geq 0$, select a subset of $s_0$ vertices of $A$ and form a clique between these $s_0$ vertices. Suppose we are now given a graph $A$ generated according to the above planted clique model with {\em unknown} $s_0$ and a positive integer $k \geq 1$. The $\mathrm{Planted Clique}$ problem seeks to determine if $A$ contains a clique of size at least $k$. It is conjectured that there is no polynomial time algorithm for the $\mathrm{Planted Clique}$ problem that works for {\em all} values of $k \in \{1,2,\dots,n\}$, unless the $\mathrm{P} = \mathrm{NP}$ hypothesis in computational complexity holds \citep{braverman2017eth}. In particular, if $\log n \ll k \ll n^{1/2}$, then all known algorithms require quasi-polynomial time of  $n^{O(\log n)}$ \citep{barak2019nearly,alon1998finding}.

Let $(A, B)$ be $R$-correlated Erd\H{o}s-R\'{e}nyi graphs with marginal edges probability $p = \tfrac{1}{2}$ and correlation matrix $R$ constructed as follows. First, select a subset $\mathcal{C}$ of vertices with $|\mathcal{C}| = s_0$. Then define
    \begin{equation}
      \label{eq:planted}
    	R_{ij} =
   		\begin{cases}
   			-1,	&	\text{if $i \in \mathcal{C}$ and $j \in \mathcal{C}$},\\
   			0,	&	\text{otherwise}.
   		\end{cases}
    \end{equation}
Note that $\|R\|_{F} = s_0$ and furthermore, $(A_{ij}, B_{ij}) \in \{(0,1), (1,0)\}$ whenever $i,j \in \mathcal{C}$ and $i \neq j$. Now consider testing $\mathbb{H}_0 \colon \|R\|_{F} = 0$ against $\mathbb{H}_A \colon \|R\|_{F} > 0$. Assume $s_0$ is unknown with $\log n \ll s_0 \ll n^{1/2}$ and suppose that there exists a polynomial time consistent test procedure for testing $\mathbb{H}_0$ and $\mathbb{H}_A$. Let $S = |A - B|$ where the absolute value is taken elementwise, i.e., $S_{ij} = |A_{ij} - B_{ij}|$. Note that $S$ is the adjacency matrix of a random graph generated from the planted clique model with clique size $s_0$ and edges probability $p = \tfrac{1}{2}$. Then a consistent test procedure for testing $\mathbb{H}_0$ against $\mathbb{H}_A$ will also yield a
polynomial time algorithm for deciding whether or not $S$ has a clique of size at least $s_0$, thereby contradicting the claim of the Planted Clique conjecture; see the appendix for a more formal argument. In summary we have the following example of a statistical versus computational gap for independence testing.

\begin{mythm}\label{thm:gap}
  Let $(A, B)$ be a pair of $R$-correlated Erd\H{o}s-R\'{e}nyi graphs on $n$ vertices where $R$ is of the form in Eq.~\eqref{eq:planted} for some $\mathcal{C}$ with $\log n \ll |\mathcal{C}| \ll n^{1/2}$. Then, assuming the Planted Clique conjecture holds, there is no polynomial time consistent test procedure for testing $\mathbb{H}_0 \colon \|R\|_{F} = 0$ against $\mathbb{H}_A \colon \|R\|_{F} > 0$ even though $\|R\|_{F} \rightarrow \infty$ as $n \rightarrow \infty$.  
\end{mythm}

\section{Independence testing in the graphon model}
\label{graphon Model}
\subsection{Same Marginal Distribution}\label{general}
We now describe independence testing when $A$ and $B$ are, {\em marginally}, generated from the class of latent positions models or graphons \citep{hoff2002latent,bollobas2007phase,lovasz12:_large}. In particular, we will derive an asymptotically valid and consistent test procedure that also runs in time polynomial in $n$. We first define the notion of a pair of $R$-correlated latent position graphs where the correlation matrix $R$ is also generated from a collection of latent positions.  This is a natural extension of the latent positions model for a single graph to the setting of two graphs sharing a common vertex set with edges that are possibly pairwise correlated. 
\begin{mydef}\label{def_graphon}
Let $\{X_1, \dots, X_n\} \subset U \subset \mathbb{R}^{d}$, and $\{Y_1, \dots, Y_n\} \subset V \subset \mathbb{R}^{d'}$ be two collections of {\em latent positions}. Now let $h$ be a symmetric bivariate function from $\mathbb{R}^{d} \times \mathbb{R}^{d}$ to $[0,1]$ and $g$ be a symmetric bivariate function from $\mathbb{R}^{d'} \times \mathbb{R}^{d'}$ to $[-1,1]$; assume also that $h$ and $g$ do not depend on $n$. Let $\rho_{n} \in [0,1]$ and $\gamma_n \in
[0,1]$ and define the $n \times n$ matrices $P$ and $R$ by 
\begin{equation*}
        P_{ij}	=	\rho_{n} \cdot h(X_i, X_j), \quad R_{ij}	=   \gamma_{n} \cdot g(Y_i, Y_j)
 \end{equation*}
where we have implicitly assumed that $\gamma_n$ and $g$ are chosen so that $R_{ij}$ satisfies the constraint in  Eq.~\eqref{eq:r_ij_bound1} for all $i,j$. Given $P$ and $R$, we generate $(A, B)$ as in Definition~\ref{def_1}. We then say that $(A, B)$ is a pair of correlated latent position graphs with correlation matrix $R$ and edge probabilities matrix $P$. 
\end{mydef}

\begin{remark}
  In Definition~\ref{def_graphon}, the factor $\rho_n$ controls the {\em sparsity} of the observed graphs $A$ and $B$ while the factor $\gamma_n$ controls the magnitudes of the pairwise correlations $\{R_{ij}\}$. Note that the average degrees for both $A$ and $B$ is $\Theta(n \rho_n)$; in the following presentation, we will allow for the possibility that $\rho_n \rightarrow 0$ and $\gamma_n \rightarrow 0$ as $n \rightarrow \infty$.
\end{remark}
 
Let $(A, B)$ be a pair of graphs generated according to Definition~\ref{def_graphon} and $C$ be the matrix with entries
\begin{equation}
  \label{eq:Cmatrix}
    C_{ij} = \begin{cases} 1 & \text{if $A_{ij} + B_{ij} > 0$} \\
    0 & \text{otherwise} \end{cases}
\end{equation}
Denote by $H$ the matrix whose entries are $H_{ij} = \mathbb{E}[C_{ij}]$, i.e.,
\begin{equation}
  \label{eq:graphon_H}
    \begin{split}
    	H_{ij} &=	2P_{ij} - P_{ij}^2 - R_{ij} P_{ij}(1 - P_{ij})  \\ &= \rho_n h(X_i,
        X_j)  + (1 -  \gamma_n g(X_i, X_j)) \rho_n h(X_i, X_j) (1 - \rho_n h(X_i, X_j)).
    \end{split}
\end{equation}
We now consider testing $\mathbb{H}_{0} \colon \|R\|_{F} = 0$ against $\mathbb{H}_{A} \colon \|R\|_{F} > 0$. Our test statistic is based on estimating $R$ using singular value thresholding (USVT). More specifically, we first compute an estimate $\hat{P}_1$ (resp. $\hat{P}_2$) of $P$ by truncating the singular value decomposition of $A$ (resp. $B$) to keep only the $k_A$ (resp. $k_B$) largest singular values. Here $k_A := \{\max k \colon \sigma_k(A) \geq c_0 \sqrt{n \rho_n}\}$, the $\sigma_1(A) \geq \sigma_2(A) \geq \dots \geq 0$ are the singular values of $A$, and $c_0 > 4$ is a universal constant; the value $k_B$ is defined analogously. See \cite{chatterjee2015matrix,xu2017rates} for more details. We also apply the same operations to $C$ and obtain an estimate $\hat{H}$ of $H$. Let $\circ$ denote the Hadamard product for matrices. Then under $\mathbb{H}_0$ we have $\|H - 2P + P \circ P\|_{F} = 0$ and thus we consider a test statistic based on $\|\hat{H} - 2 \hat{P} + \hat{P} \circ \hat{P}\|_{F}$ where $\hat{P} = (\hat{P}_1 + \hat{P}_2)/2$. Leveraging recent results on the estimation error of SVT for graphon estimation \citep{xu2017rates,chatterjee2015matrix}, we obtain the following consistency guarantee for our test procedure. 

\begin{mythm}\label{theorem 1}
    Let $(A, B)$ be a pair of graphs generated according to the model in Definition~\ref{def_graphon}, where $g$ and $h$ are fixed functions and do not vary with $n$. Assume that both $g$ and $h$ are at least $s$ times differentiable for some $s \geq 1$, where $s$ is assumed known, and that $n \rho_n = \omega(\log n)$ as $n$ increases. Let $\alpha = \tfrac{s + d + d'}{2s + d + d'}$ and define 
    $$T(A, B) = \frac{\|\hat{H}
    - 2 \hat{P} + \hat{P} \circ \hat{P}\|_{F}}{\Delta^{\alpha} \log^{1/2}{n}}.$$
    where $\Delta$ is the average of the maximum degree of $A$ and $B$. Let $\mathcal{R} = \{T \colon T > 1\}$. The test statistic $T(A, B)$ with rejection region $\mathcal{R}$ yields an asymptotically valid test procedure for testing $\mathbb{H}_0 \colon \|R\|_{F} = 0$ against $\mathbb{H}_A \colon \|R\|_{F} > 0$ and furthermore, $T$ is consistent whenever $\|R \circ (P - P \circ P)\|_{F} = \Omega((n \rho_n)^{\alpha'})$  for any $\alpha' > \alpha$ as $n \rightarrow \infty$. 
\end{mythm}
We note that the assumption $n \rho_n = \omega(\log n)$ in Theorem~\ref{theorem 1} guarantees that the average degrees of the observed graphs grow slightly faster than $\log n$. This then allows the singular value thresholding step to yield reasonably accurate estimates $\hat{P}$ and $\hat{H}$. See for example the condition in Eq.~3 of \citet{xu2017rates}. Furthermore, the consistency regime for $T$ in Theorem~\ref{theorem 1} is stated in terms of $\|R \circ (P - P \circ P)\|_{F}$ as opposed to $\|R\|_{F}$. This is expected as the entries of $R \circ (P - P \circ P)$, which are $R_{ij} P_{ij} (1 - P_{ij})$, correspond to the difference between the edge probabilities matrix of $(A,B)$ under $\mathbb{H}_0$ and $\mathbb{H}_A$. If we assume that the link function $h$ satisfies $h(x,x') > 0$ for all $(x, x') \in U \times U$ then $P_{ij} = \Omega(\rho_n)$ for all $i,j$, and the condition $\|R \circ (P - P \circ P)\|_{F} = \Omega((n \rho_n)^{\alpha'} )$ simplifies to $\rho_n \|R\|_{F} = \Omega((n \rho_n)^{\alpha'} )$, or equivalently that $\|R\|_{F} = \Omega(n^{\alpha'} \rho_n^{\alpha' - 1})$. As $\alpha$ decreases when $s$ increases, we see that smoother $h$ and $g$ lead to a sharper consistency threshold for $\|R\|_{F}$. As a special case of Theorem~\ref{theorem 1}, suppose that both $g$ and $h$ are {\em infinitely} differentiable. We then have $\alpha \leq 1/2 + \epsilon$ for any $\epsilon > 0$ and our test procedure is consistent whenever $\|R \circ (P - P \circ P)\|_{F} = \Omega((n\rho_n)^{1/2 + \epsilon})$ for any $\epsilon > 0$. More specifically we have the following corollary.
  
\begin{corollary}
    \label{cor:lpg}
    Consider the setting in Theorem~\ref{theorem 1} and suppose that $g$ and $h$ are both infinitely differentiable. Now choose an arbitrary $\epsilon > 0$ not depending on $n$ and define 
    $$T(A, B) = \frac{\|\hat{H} - 2 \hat{P} + \hat{P} \circ \hat{P}\|_{F}}{\Delta^{1/2 + \epsilon/2}}.$$
    Let the rejection region be given by $\mathcal{R} = \{T \colon T >   1\}$. Then the test statistic $T(A, B)$ with rejection region $\mathcal{R}$ yields an asymptotically valid test procedure for testing $\mathbb{H}_0 \colon \|R\|_{F} = 0$ against $\mathbb{H}_A \colon \|R\|_{F} > 0$ and furthermore, $T$ is consistent whenever $\|R \circ (P - P \circ P)\|_{F} =  \Omega((n \rho_n)^{1/2 + \epsilon})$ as $n \rightarrow \infty$.  
\end{corollary}

\begin{remark}
  If we suppose that (1) $h(x,x') > 0$ for all $x,x' \in U$ and (2) $g(y, y') > 0$ for all $y, y' \in V$ then $\|R \circ (P - P \circ P)\|_{F} = \Theta(n \gamma_n \rho_n )$ and the condition $\|R \circ (P - P \circ P)\|_{F} = \Omega((n \rho_n)^{1/2 + \epsilon})$ in Corollary~\ref{cor:lpg} is equivalent to the condition $\gamma_n = \Omega((n \rho_n)^{-1/2 + \epsilon})$, i.e., the test procedure is consistent whenever the pairwise correlations decay to $0$ slower than the reciprocal of the square root of the average degree. 
\end{remark}
  
\begin{small}
\begin{algorithm}[tp]  
  \caption{Bootstrap procedure for graphons}  
  \label{Graphon_Model_Testing_Simulation_2}  
  \begin{algorithmic}
    \Require Adjacency matrices $A$ and $B$, both of size $n \times n$, significance level $\alpha \in (0,1)$, number of bootstrap samples $m$.
     \State (A) Compute the matrix $C$ whose entries are $C_{ij} = 1$ if $A_{ij} + B_{ij} > 0$ and $C_{ij} = 0$ otherwise. 
    \State (B) Compute $\hat{P}_1, \hat{P}_2$ and $\hat{H}$ by applying universal singular value thresholding (USVT) on $A$, $B$, and $C$, respectively. 
    \State (C) Let $\hat{P}=\frac{1}{2}(\hat{P}_1+\hat{P}_2)$ and calculate the test statistic $T = \|\hat{H} - 2 \hat{P} + \hat{P} \circ \hat{P}\|_{F}$.
    \For{$s=1$ to $m$}
       \State (i) Generate adjacency matrices $(A^{(s)}, B^{(s)})$ according to Definition~\ref{def_1} with $R = 0$ and marginal edge probabilities matrix $\hat{P}$. 
        \State (ii) Compute $\hat{P}_1^{(s)}$ and $\hat{P}_2^{(s)}$ as the universal singular value threshold of $A^{(s)}$ and $B^{(s)}$, respectively. 
       \State (iii) Calculate $T^{(s)} = \|\hat{H}^{(s)}- 2 \hat{P}^{(s)} + \hat{P}^{(s)} \circ \hat{P}^{(s)}\|_{F}$, where $\hat{P}^{(s)}=\frac{1}{2}(\hat{P}_1^{(s)}+\hat{P}_2^{(s)})$ and $\hat{H}^{(s)}$ is the universal singular value thresholding of $A^{(s)} + B^{(s)}$.
      \EndFor
      \State (D) Set $c_{\alpha}$ to be the $(1 - \alpha) \times 100\%$ percentile of the $\{T^{(s)}\}_{s=1}^{m}$
      \State \textbf{Output} If $T>c_{\alpha}$ then reject $\mathbb{H}_0$; otherwise fail to reject $\mathbb{H}_0$. 
  \end{algorithmic}  
\end{algorithm}
\end{small}

If $g$ and $h$ are not infinitely differentiable then the test statistic in Theorem~\ref{theorem 1} depends on knowing (1) a lower bound for the smoothness $s$ of the functions $g$ and $h$ and (2) upper bounds for the dimensions $d$ and $d'$ of the latent positions $\{X_i\}$ and $\{Y_i\}$. These values are most likely unknown in practice. Furthermore, even when $s, d$ and $d'$ are known, for finite sample the rejection region in Theorem~\ref{theorem 1} is likely to be overly conservative. This is a common issue in many graph testing problems whose test statistics have no known non-degenerate limiting distribution; see for example the test statistics in \citet{tang2017semiparametric}, \citet{ghoshdastidar2020two}, and \citet{gretton2012kernel}. In light of these limitations, for this paper we will instead consider the {\em unnormalized} test statistic 
$\tilde{T}(A, B) = \|\hat{H} - 2 \hat{P} + \hat{P} \circ \hat{P}\|_{F}$ and use a bootstrap procedure to determine the rejection region for $\tilde{T}$. More specifically the bootstrap procedure generates additional pairs of graphs $\{A_b, B_b\}_{b=1}^{B}$ from the estimated edge  probabilities matrix $\hat{P}$ with $R = 0$ and then computes the test statistics for these bootstrap pairs to obtain an empirical distribution for the test statistics under the null hypothesis. See Algorithm~\ref{Graphon_Model_Testing_Simulation_2} for more details. Note that $\tilde{T}(A, B)$ differs from $T(A, B)$ only in the term $\Delta^{\alpha} \log^{1/2}{n}$. This term is a normalizing factor that guarantees $T(A, B) < 1$ when $\|R\|_{F} = 0$ and $T(A, B) \rightarrow \infty$ when $\|R\|_{F} = \Omega((n \rho_n)^{\alpha'})$. Hence, by using $\tilde{T}(A, B)$ and bootstrapping, we circumvent the need to know/estimate $\alpha$ (a possibly non-trivial if not impossible task).

\subsection{Detection thresholds for stochastic blockmodels}
\label{sec:SBM}
We now consider the special case of independence testing when the graphs $A$ and $B$ are, {\em marginally}, stochastic blockmodel graphs with a common block structure. We begin by formulating an extension of the stochastic blockmodel \citep{holland1983stochastic} for a single graph to the case of a pair of graphs whose edge correlations also exhibit a block structure. This extension had appeared previously in the literature in the context of graph matching (see e.g., \cite{onaran2016optimal,racz2021correlated,lyzinski2014seededER}).

\begin{mydef}
  \label{def:sbm}
  Let $A$ and $B$ be graphs on $n$ vertices. We say that $(A, B)$ is a pair of $K$-blocks correlated stochastic blockmodel graphs with common community assignments $\tau$, block probabilities matrices $\Theta_P$ and $\Theta_Q$, sparsity factor $\rho_n$ and block correlations matrix $\Theta_R$ if $(A,B)$ are generated from the $R$-$\mathrm{ER}(P,Q)$ model where
  \begin{enumerate}
    \item $\Theta_P$ and $\Theta_Q$ are $K \times K$ {\em symmetric} matrices with entries in $[0,1]$.
    \item Marginally, the edge probabilities matrix for $A$ and $B$ are $P = \rho_n Z \Theta_P Z^{\top}$ and $Q = \rho_n Z \Theta_Q Z^{\top}$ where $Z$ is a $n \times K$ matrix such that, for any $k \in \{1,2,\dots,K\}$, we have $Z_{ik} = 1$ if $\tau_i = k$ and $Z_{ik} = 0$ otherwise. 
    \item The correlation matrix $R$ is $R = Z \Theta_R Z^{\top}$. 
  \end{enumerate}
\end{mydef}
Let $(A, B)$ be generated from the model in Definition~\ref{def:sbm}. We now describe a test statistic $T$ for testing $\mathbb{H}_0 \colon \|R\|_{F} = 0$ against $\mathbb{H}_A \colon \|R\|_{F} > 0$ with the following properties (1) $T$ has a central $\chi^2$ limiting distribution under the null hypothesis and (2) $T$ is consistent under the alternative hypothesis provided that $\|R\|_{F} \rightarrow \infty$.

First, compute the matrix $C$ as defined in Eq.~\eqref{eq:Cmatrix} and cluster the vertices of $C$ into $K$ communities using a community detection algorithm that guarantees exact recovery (see e.g., \cite{abbe2017community,gao2017achieving,lyzinski2014perfect}) where the value of $K$ can be chosen using model selection procedures such as those described in \cite{li_network,wang_bickel_llr,lei2014}. Let $\hat{\tau}$ be the resulting estimated community assignment. Next, compute, for $1 \leq k \leq \ell \leq K$, the Pearson sample correlation $\hat{\rho}_{k \ell}$ between the edges $\{A_{ij}, B_{ij}\}$ in the $(k,\ell)$th block, i.e., $$\hat{\rho}_{k \ell} = \mathrm{cor}\Bigl(\{A_{ij} \colon i < j, \hat{\tau}_i = k, \hat{\tau}_j = \ell\}, \{B_{ij} \colon i < j, \hat{\tau}_i = k, \hat{\tau}_j = \ell\}\Bigr).$$ Let $\hat{n}_k = |\{ i \colon \hat{\tau}_i = k\}|$ for all $k$ and define the test statistic $$T(A, B) = \sum_{k \leq \ell} \hat{n}_{k \ell} \hat{\rho}_{k \ell}^{2}$$
where $\hat{n}_{kk} = \tbinom{\hat{n}_k}{2}$ if $k = \ell$ and $\hat{n}_{k \ell} = \hat{n}_k \hat{n}_{\ell}$ if $k \not = \ell$.
We then have the following result.
\begin{mythm}
\label{thm:SBM}
Let $(A,B)$ be a pair of graphs on $n$ vertices generated according to Definition~\ref{def:sbm} where $\Theta_P$ and $\Theta_Q$ are fixed $K \times K$ matrices not depending on $n$. Suppose that, as $n$ increases, we have $n_k = |\{i \colon \tau_i = k\}| = \Theta(n)$ for all $k \in\{1,2,\dots,K\}$ and the sparsity $\rho_n$ satisfies $n \rho_n = \omega(\log n)$. Then under $\mathbb{H}_0 \colon \|R\|_{F} = 0$ we have $T(A, B) \rightsquigarrow \chi^2_{K(K+1)/2}$ as $n \rightarrow \infty$. Furthermore let $\mu > 0$ be a finite constant such that $\|R\|_{F} > 0$ satisfies
\begin{equation}
\label{eq:local_alternative}
\sum_{k \leq \ell} n_{k \ell} \Theta_R^2(k,\ell) \rightarrow \mu
\end{equation}
as $n \rightarrow \infty$ (here $\Theta_{R}(k,\ell)$ denote the $k,\ell$th entry of $\Theta_R$). We then have
$$T(A, B) \rightsquigarrow \chi^2_{K(K+1)/2}(\mu)$$ as $n \rightarrow \infty$. Here $\chi^2_{K(K+1)/2}(\mu)$ is a non-central $\chi^2$ with $K(K+1)/2$ degrees of freedom and non-centrality parameter $\mu$. 
\end{mythm}
Theorem~\ref{theorem 1} (and its associated Corollary~\ref{cor:lpg}) to Theorem~\ref{thm:SBM} we see that by assuming a more specialized structure on $R$ we obtain a much sharper detection threshold. Indeed, the local alternative in Theorem~\ref{thm:SBM} implies that our test statistic achieves power converging to $1$ whenever $\Theta_R$ satisfies the condition 
$\sum_{k \ell} n_{k \ell} \Theta_R^2(k,\ell) \rightarrow \infty$. This is equivalent to the condition that $\|R\|_{F} \rightarrow \infty$ and is thus, by Theorem~\ref{thm:main1}, both necessary and sufficient. Theorem~\ref{thm:SBM} also assumes that the block structure for $R$ is the same as that for $P$ and $Q$, and this is done purely for ease of exposition as it allows us to more easily aggregate the edges of $A$ and $B$ into blocks and estimate the pairwise correlations between the edges in each block. Extending Theorem~\ref{thm:SBM} to the case where the SBM structure for $R$ differs from that of $P$ and $Q$ is straightforward but tedious, and we leave it to the interested reader. Finally, \cite{xiong2019graph} considered a special case of Definition~\ref{def:sbm} with $R_{ij} \equiv c$ for all $\{i,j\}$, but they did not derive a non-degenerate limiting distribution for their test statistic.

\subsection{Different Marginal Distributions}\label{different}
We now discuss how the previous model and results can be extended to the case where the graphs $A$ and $B$ have different {\em marginal} distributions. We first define a model that generalizes the $R$-$\mathrm{ER}(P)$ model in Definition~\ref{def_1} (see also Remark~\ref{remark_ext}).
\begin{mydef}
  Let $n \in \mathbb{N}$. Let $P \in [0,1]^{n\times n}$, $Q \in [0,1]^{n \times n}$, and $R\in[{-1},1]^{n\times n}$ be symmetric matrices where $R$ satisfies the constraint
  \begin{equation}
  \label{eq:constraint_different}
      -\min\Bigl\{\tfrac{P_{ij}Q_{ij}}{(1-P_{ij})(1-Q_{ij})},\tfrac{(1-P_{ij})(1-Q_{ij})}{P_{ij}Q_{ij}}\Bigr\}^{1/2} \leq R_{ij}\leq\min\Bigl\{\tfrac{P_{ij}(1-Q_{ij})}{Q_{ij}(1-P_{ij})},\tfrac{Q_{ij}(1-P_{ij})}{P_{ij}(1-Q_{ij})}\Bigr\}^{1/2},
  \end{equation}
  for all $i,j$. We say that $(A,B)$ are $R$-correlated heterogeneous Erd\H{o}s-R\'{e}nyi graphs on $n$ vertices with {\em marginal} edge probabilities $(P,Q)$ and correlations $R$, denoted by $(A,B)\sim R$-$\mathrm{ER}(P,Q)$, if 
    \begin{enumerate}
        \item
            $A$ is the adjacency matrix for an inhomogeneous Erd\H{o}s-R\'{e}nyi graph on $n$ vertices with probability matrix $P$. 
        \item
            $B$ is the adjacency matrix for another inhomogeneous Erd\H{o}s-R\'{e}nyi graph on $n$ vertices with probability matrix $Q$.
        \item
            The pairs of entries $\{(A_{ij},B_{ij})\}_{1\leq i<j\leq n}$ are independent bivariate random vectors. 
        \item
            For $i < j$, $A_{ij}$ and $B_{ij}$ are correlated with Pearson correlation $R_{ij}$. 
    \end{enumerate}
\end{mydef}
Following the proof of Theorem~\ref{thm:main1}, we can show that the condition $\limsup \|R\|_{F} < \infty$ as$n \rightarrow \infty$ is also necessary for the existence of a consistent test procedure for detecting the correlation between $R$-$\mathrm{ER}(P,Q)$ graphs. A simple generative model for $R$-$\mathrm{ER}(P,Q)$ graphs is the following variant of the graphon model described in Definition~\ref{def_graphon} where, instead of only having a single $h$, we have two link functions $h_1$ and $h_2$ from $U \times U \mapsto [0,1]$ and define $P$ and $Q$ via $P_{ij} = \rho_n h_1(X_i, X_j)$ and $Q_{ij} = \rho_n h_2(X_i, X_j)$ for all $i \leq j$. The function $g$ is also chosen so that the correlation matrix $R$ with entries $R_{ij} = \gamma_n g(Y_i, Y_j)$ satisfies the constraint in Eq.~\eqref{eq:constraint_different}. Note that for ease of exposition we had assumed that the sparsity factor for both $P$ and $Q$ are the same. The case where $P$ and $Q$ have different sparsity factors involves more tedious book-keeping but is, otherwise, conceptually identical and leads to similar results as those described below.  

Given $(A, B)$ sampled from the above variant of the $R$-$\mathrm{ER}(P,Q)$ model we once again let $C$ be the matrix with entries $C_{ij} = 1$ if $A_{ij} + B_{ij} > 0$ and $C_{ij} = 0$ otherwise. Denote $H = \mathbb{E}[C]$. We then have
\begin{equation}
  \label{eq:h_form}
  H_{ij} = P_{ij} + Q_{ij} - P_{ij} Q_{ij} -
  R_{ij}\sqrt{P_{ij}(1 - P_{ij})Q_{ij}(1 - Q_{ij})}
\end{equation}
and thus we can construct a test statistic for $\mathbb{H}_0 \colon \|R\|_{F}$ against $\mathbb{H}_A \colon \|R\|_{F} > 0$ by first applying USVT to $A$ and $B$ to obtain estimates $\hat{P}$ of $P$ and $\hat{Q}$ of $Q$, then applying USVT to $C$ to obtain an estimate $\hat{H}$ of $H$, and finally compute $\tilde{T}(A, B) = \|\hat{H} - \hat{P} - \hat{Q} + \hat{P} \circ \hat{Q}\|_{F}.$ The following result is derived using an identical argument to that for Theorem~\ref{theorem 1} and shows that rejecting $\mathbb{H}_0$ for large values of $\tilde{T}$ yields a consistent test procedure.

\begin{mythm}\label{theorem 4}
    Let $(A, B)$ be a pair of graphs generated from the above $R$-$\mathrm{ER}(P,Q)$ model where $g$, $h_1$ and $h_2$ are fixed functions and do not vary with $n$. Assume that all of $g, h_1, h_2$ are at least $s$ times continuously differentiable for some $s \geq 1$, where $s$ is assumed known, and that $n \rho_n = \omega(\log n)$ as $n$ increases. Let $\alpha = \tfrac{s + d + d'}{2s + d + d'}$ and define 
    $$T(A, B) = \frac{\|\hat{H}
    - \hat{P} - \hat{Q} + \hat{P} \circ \hat{Q}\|_{F}}{\Delta^{\alpha} \log^{1/2}{n}}.$$
     where $\Delta$ is the average of the maximum degree of $A$ and $B$. Let $\mathcal{R} = \{T \colon T > 1\}$. The test statistic $T(A, B)$ with rejection region $\mathcal{R}$ yields an asymptotically valid test procedure for testing $\mathbb{H}_0 \colon \|R\|_{F} = 0$ against $\mathbb{H}_A \colon \|R\|_{F} > 0$ and furthermore, $T$ is consistent whenever $\|R \circ \Xi\|_{F} = \Omega((n \rho_n)^{\alpha'})$ for any $\alpha' > \alpha$ as $n \rightarrow \infty$. Here $\Xi$ is the $n \times n$ matrix with entries $\Xi_{ij} = \bigl(P_{ij}(1 - P_{ij}) Q_{ij}(1 - Q_{ij})\bigr)^{1/2}.$
\end{mythm}
Similar to our discussion after Corollary~\ref{cor:lpg}, if the link functions $g$, $h_1$ and $h_2$ are not infinitely differentiable then the
test statistic in Theorem~\ref{theorem 4} depends on knowing a lower bound for the smoothness $s$ and upper bounds for the dimensions $d$ and $d'$ of the latent positions $\{X_i\}$ and $\{Y_i\}$. These values are once again unknown or, even if known, the rejection region in Theorem~\ref{theorem 4} is still overly conservative for finite sample inference. We will thus also use bootstrap resampling to calibrate
the rejection region for the above test statistic $T(A, B)$; see Algorithm~\ref{Analysis_for_real_data} in Section~\ref{Real_Data} for more details.

\section{Simulation Results}
We now conduct simulation experiments to evaluate the performance of our test procedures for testing the hypothesis in Section~\ref{graphon Model}. For our first experiment, we generate $\{X_i\}_{i=1}^{n}$ as iid sample from a bivariate normal with mean $\bm{0}$ and identity covariance matrix. We then consider two different choices of link functions for $P$. The first is the cosine similarity 
\begin{equation}
\label{eq:cosine}
    P_{ij} = \frac{|X_i^{\top}X_j|}{2\|X_i\|\|X_j\|}
\end{equation}
and the second is the Gaussian kernel $P_{ij} = \exp\big(-{\|X_i-X_j\|^2}\big)/2$. Note that the rank of $P$ is $2$ and $n$ for the cosine and Gaussian similarity, respectively. Given $P$ we set $R_{ij} \equiv r$ for some value $r$ to be specified
later. We then generate a pair of graphs $(A, B) \sim R$-$\mathrm{ER}(P)$ on $n$ vertices and apply the test statistic $T$ in
Corollary~\ref{cor:lpg} to test $\mathbb{H}_0 \colon \|R\|_{F} = 0$ against $\mathbb{H}_A \colon \|R\|_{F} > 0$ where the rejection region is calibrated via bootstrapping with significance level $0.05$ (see Algorithm~\ref{Graphon_Model_Testing_Simulation_2}).  We repeat the above steps for $m = 100$ Monte Carlo replicates to obtain an empirical estimate of the type I error (when $r = 0$) and type II error (when $r > 0$) for $T$. The results are presented in Table \ref{Table:same cosine} for combinations of $n \in \{100,200,500,1000,2000\}$ and $r \in \{0,0.1,0.3,0.5\}$. We observe that the type I error of the test statistic is well-controlled and that the test statistic also exhibits power even for small values of $r$ and moderate values of $n$. Note that we fixed $k = 2$ when $P$ is the cosine similarity; here $k$ is the number of singular values used to construct $T$. In contrast, we chose $k$ via USVT (see Section~\ref{graphon Model}) when $P$ is the Gaussian similarity. 
\begin{table}[tp]
\centering
\begin{tabular}{lcccc}
\hline
$n \setminus r$ &$r=0$&$r=0.1$&$r=0.3$&$r=0.5$ \\
    \hline
    $n=100$ & 0.06 & 0.17 & 0.85 & 1\\
    $n=200$ & 0.07 & 0.98 & 1 & 1 \\
    $n=500$ & 0.04 & 1    & 1 & 1\\
    $n=1000$ & 0.02 & 1    & 1 & 1\\
    $n=2000$ & 0.01 & 1    & 1 & 1\\
    \hline

\end{tabular}
\caption{Empirical estimates (based on $m = 100$ Monte Carlo replicates) for the Type I and type II error when $P_{ij}$ is from the cosine similarity. Given values correspond to the Type-I error when $r = 0$ and to the power (i.e., one minus the Type II error) when $r > 0$.}
\label{Table:same cosine}
\end{table}

\begin{table}[tp]
\centering
\begin{tabular}{lcccc}
\hline
$n \setminus r$ &$r=0$&$r=0.1$&$r=0.3$&$r=0.5$ \\
    \hline
    $n=100$  & 0 & 0    & 0.06 & 0.74\\
     $n=200$ & 0 & 0    & 0.26 & 1\\
    $n=500$  & 0 & 0    & 1 &1\\
     $n=1000$  & 0 & 0  & 1 & 1\\
     $n=2000$ & 0 & 0.02 & 1 & 1 \\
     \hline
\end{tabular}

\caption{Empirical estimates (based on $m = 100$ Monte Carlo replicates) for the type I and type II error when $P_{ij}$ is from the Gaussian kernel. The given values correspond to the Type-I error when $r = 0$ and to the power (i.e., one minus the type II error) when $r > 0$.}
\label{Table: exponential}
\end{table}

For our second experiment we consider the case where the graphs $A$ and $B$ have different marginal distributions (see Section~\ref{different}). In particular we generate $\{X_i\}_{i=1}^{n}$ and $\{Y_i\}_{i=1}^{n}$ as iid samples from a bivariate normal with mean $\bm{0}$ and identity covariance matrix and then set $P$ and $Q$ to have entries
\begin{gather*}
    P_{ij} = \frac{|X_i^{\top} X_j|}{2 \|X_i\|\|X_j\|}, \quad \text{and} \quad Q_{ij} = \frac{|Y_i^{\top} Y_j|}{4\|Y_i\|\|Y_j\|}.
\end{gather*}
Given $P$ and $Q$ we once again set $R_{ij} \equiv r$ and generate pair of graphs $(A, B)$ on $n$ vertices from the $R$-$\mathrm{ER}(P,Q)$ model. We apply the test statistic in Theorem~\ref{theorem 4} to test $\mathbb{H}_0 \colon \|R\|_{F} = 0$ against $\mathbb{H}_A \colon \|R\|_{F} > 0$; the rejection region is also calibrated via bootstrapping with significance level $0.05$ using Algorithm~\ref{Analysis_for_real_data}. Empirical estimates of the type I and type II errors (based on $m = 100$ Monte Carlo replicates) are presented in Table \ref{Table:Different} for combinations of $n \in \{100,200,500,1000,2000\}$ and $r \in \{0,0.1,0.3,0.5\}$. Once again, the type-I error is well-controlled and the test also exhibits power even for small values of $r$ and moderate values of $n$.

\begin{table}[tp]
\centering
\begin{tabular}{lcccc}
\hline
$n \setminus r$ &$r=0$&$r=0.1$&$r=0.3$&$r=0.5$ \\
    \hline
    $n=100$ & 0 & 0.35 & 1 & 1\\
    $n=200$ & 0 & 0.83 & 1 & 1 \\
    $n=500$ & 0.11 & 1    & 1 & 1\\
    $n=1000$ & 0.01 & 1    & 1 & 1\\
    $n=2000$ & 0.02 & 1    & 1 & 1\\
    \hline
\end{tabular}
\caption{Empirical estimates (based on $m = 100$ Monte Carlo replicates) for the type $I$ and type $II$ error when $P_{ij}$ and $Q_{ij}$ are from the cosine similarity with $P \not = Q$. The given values correspond to the type I error when $r = 0$ and to the power (i.e., one minus the type II error) when $r > 0$.}
\label{Table:Different}
\end{table}

For our last experiment, we consider correlation testing when $A$ and $B$ are {\em marginally} stochastic blockmodel graphs (c.f. Section~\ref{sec:SBM}). More specifically, we assume that $A$ is a $K$-block SBM where each vertex of $A$ is assigned to some a block $k \in \{1,\dots,K\}$ with probability $1/K$ and the marginal edge probabilities between vertices in block $k$ and vertices in block $\ell$ are given by $0.45 - |k - \ell|/(2K)$ for any $k, \ell \in \{1,2,\dots,K\}$. Similarly, $B$ is a $K$-block SBM with the same membership assignment as $A$ and marginal edge probabilities $0.4 - |k - \ell|/(2K+2)$. We set $R$ to have the same block structure as both $A$ and $B$ and entries of the form $r( 1 - |k - \ell|/K)$ for values of $r$ that are specified later.  We then sample a pair $(A,B)$ from the model in Definition~\ref{def:sbm} with parameters given above and test the hypothesis that $\|R\|_{F} = 0$ against $\|R\|_{F} > 0$ using the test statistic in Theorem~\ref{thm:SBM} where the rejection region is based on the $95\%$ percentile of the $\chi^2_{K(K+1)/2}$ distribution. We repeat these steps for $m = 1000$ Monte Carlo replicates to obtain empirical estimates of the type-I and type-II error of our test statistic. The results are presented in Tables \ref{Table: SBM1} for various combinations of $n \in \{200,500, \dots, 3000\}$, $K \in \{2,5,7\}$ and $r \in \{0,0.001,0.005,0.01\}$. For comparisons we also include
the limiting theoretical power given by $F^{-1}_{\mu}(c_*)$ where $c_*$ is the $95\%$ percentile of the (central) $\chi^2_{K(K+1)/2}$ distribution and $F^{-1}_{\mu}$ is the quantile function for a {\em non-central} $\chi_{K(K+1)/2}^2$ with non-centrality parameter $\mu = r^2(\tfrac{K^2+1}{4K^2} n^2 - \tfrac{n}{2})$. We observe that the empirical type-I errors (when $r = 0$) and empirical power (when $r > 0$) are close to their limiting theoretical counterparts provided that $n$ is sufficiently large compared to $K$; indeed, as $K$ increases we generally need larger values of $n$ to achieve accurate recovery of the latent community assignments. 

\begin{table}[htp]
\centering
\label{Table: SBM1}
\begin{tabular}{lcccc}
\hline
 $K = 2$ &$r=0$&$r=0.001$&$r=0.005$&$r=0.01$ \\
    \hline
      $n=200$ & 0.044/0.050 & 0.045/0.051 & 0.056/0.069 & 0.147/0.133 \\
    $n=500$ & 0.047/0.050 & 0.055/0.055 & 0.177/0.188 & 0.621/0.641\\
    $n=1000$ & 0.056/0.050 & 0.070/0.069 & 0.638/0.642 & 1/0.999\\
    $n=2000$ & 0.048/0.050 & 0.139/0.134 & 0.998/0.999 & 1/1\\
    $n=3000$& 0.040/0.050 & 0.236/0.259 & 1/1& 1/1\\
    \hline
\end{tabular}

\begin{tabular}{lcccc}
\hline
 $K=5$&$r=0$&$r=0.001$&$r=0.005$&$r=0.01$ \\
    \hline
    $n=200$ & 0.864/0.050 & 0.755/0.050 & 0.795/0.056 & 0.888/0.076 \\
    $n=500$ & 0.055/0.050 & 0.053/0.051 & 0.103/0.093 & 0.311/0.289\\
    $n=1000$ & 0.051/0.050 & 0.067/0.056 & 0.280/0.290 & 0.940/0.931\\
    $n=2000$ & 0.054/0.050 & 0.070/0.076 & 0.936/0.932 & 1/1\\
    $n=3000$ & 0.041/0.050 & 0.111/0.116 & 1/1 & 1/1\\
    \hline
\end{tabular}

\begin{tabular}{lcccc}
\hline
 $K=7$&$r=0$&$r=0.001$&$r=0.005$&$r=0.01$ \\
    \hline
   $n=200$ & 0.954/0.050 & 0.912/0.050 & 0.955/0.054 & 0.973/0.067 \\
    $n=500$ & 0.617/0.050 & 0.560/0.051 & 0.656/0.079 & 0.826/0.208\\
    $n=1000$ & 0.058/0.050 & 0.062/0.054 & 0.218/0.209 & 0.859/0.833\\
    $n=2000$& 0.056/0.050 & 0.074/0.068 & 0.823/0.834 & 1/1\\
    $n=3000$& 0.049/0.050 & 0.104/0.094 & 1/0.999 & 1/1\\
    \hline
\end{tabular}
\caption{Empirical estimates (based on $m = 100$ Monte Carlo replicates) for the type $I$ and type $II$ error compared to the theoretical (limiting) value. Here $A$ and $B$ are $R$-correlated SBM graphs. The first (resp. second) entry in each cell correspond to the empirical estimate (resp. theoretical value) of the type I error when $r = 0$ and to the power (i.e., one minus the type II error) when $r > 0$. The theoretical values are based on the non-central $\chi^2$ distribution with non-centrality parameter $\mu=r^2(\frac{K^2+1}{4K^2}n^2-\frac{n}{2})$.}
\end{table}

\section{Real Data Experiments}\label{Real_Data}
\subsection{Analysis of C. elegans Data}
\label{sec:c_elegans}
We now apply the test statistics in Section~\ref{graphon Model} to the connectomes of the {\em C. elegans} roundworm. More specifically, we used the wiring diagram formed by the somatic nervous system, which consists of 
$279$ neurons; these neurons are classified into one of three categories namely motor neurons, sensory neurons, and inter-neuron. There are two types of connections between the neurons, i.e., either via chemical synapses or electrical gap junctions. This result in two related but distinct networks, namely a chemical synapse network $A_c$ with $6394$ edges and a gap junction network $A_g$ with $1777$ edges. See \cite{varshney2011structural} for a more detailed description of the construction of these connectomes.

We first consider testing the null hypothesis that $A_c$ and $A_g$ are independent against the alternative hypothesis that they are correlated. As the two graphs have a quite large
difference in the edge densities, we suppose that $A_c$ and $A_g$ are generated from the $R$-$\mathrm{ER}(P,Q)$ model (see Section~\ref{different}) and use the test statistic $T$ in
Theorem~\ref{theorem 4} with $k = 3$ as the rank for the estimates $\hat{P}$ and $\hat{Q}$; the choice $k = 3$ is motivated by the fact that there are three categories of neurons. This result in an observed value of $T = 8.313$. We calibrate our test statistic using the bootstrapping procedure in Algorithm~\ref{Analysis_for_real_data} with $m = 10000$ Monte Carlo replicates and obtain an approximate $p$-value of $5 \times 10^{-5}$. There is thus strong evidence to reject the null hypothesis in favor of the alternative hypothesis that the two connectomes are correlated. This conclusion, while biologically relevant, is also certainly expected. 

We now quantify the degree of correlation between the edges of the two graphs. Recalling Eq.~\eqref{eq:h_form} we first compute an estimate of the correlation $R_{ij}$ between the edges of $A_c$ and $A_g$ via 
\begin{equation}
  \label{eq:estimate_R}
    \hat{R}_{ij}:=
    \begin{cases}
    0, & \text{if $\hat{P}_{ij}$ or $\hat{Q}_{ij}\in\{0,1\}$}\\
    \max\Big\{\min\Big\{\tfrac{\hat{P}_{ij}+\hat{Q}_{ij} -
      \hat{P}_{ij}\hat{Q}_{ij} - \hat{H}_{ij}}{\bigl(\hat{P}_{ij}(1-\hat{P}_{ij})\hat{Q}_{ij}(1-\hat{Q}_{ij})\bigr)^{1/2}},1\Big\},-1\Big\}, & \text{otherwise.}
    \end{cases}
\end{equation}
where $\hat{P}, \hat{Q}$ and $\hat{H}$ are obtained by applying USVT to the adjacency matrices $A$, $B$, and $C$ respectively; recall that $C_{ij} = 1$ if $A_{ij} + B_{ij} > 0$ and $C_{ij} = 0$ otherwise. We then compute the average correlations for edges connecting vertices from the same category as well as edges connecting vertices from different categories, e.g., we calculate the sample mean of the $\{\hat{R}_{ij}\}$ when $i$ and $j$ are both motor neurons as well as the sample mean of the $\{\hat{R}_{ij}\}$ when $i$ is a motor neuron and $j$ is a sensory neuron. The results are presented in Table \ref{Correlation_Matrix_Graphon_c.elegans} for all possible pairs of neuron categories; these correlations are all positive and quite large. 
\begin{table}[tp]
    \centering
    \begin{tabular}{|l|c|c|c|}
        \hline
        & motor & inter & sensory\\
        \hline
  motor & 0.144 & 0.111 & 0.153\\
  inter & 0.111 & 0.088 & 0.137\\
sensory & 0.153 & 0.137 & 0.193\\
        \hline
    \end{tabular}
    \caption{Sample means of the estimated correlations
      $\{\hat{R}_{ij}\}$ for different combinations of neuron types
      for $i$ and $j$.}
    \label{Correlation_Matrix_Graphon_c.elegans}
\end{table}
As a sanity check, we also compute the sample Pearson correlation based on the binary entries of $A_c$ and $A_g$ directly, i.e., we compute $\mathrm{Cor}(\{A_c(i,j), A_g(i,j)\})$ where $i$ ranges over all neurons of type $k$ and $j$ ranges over all neurons of type $\ell$ (with $k$ possibly being the same as $\ell$). The results are presented in 
Table~\ref{Correlation_Matrix_SBM_c.elegans}. We see that these sample Pearson correlations exhibit the same general pattern as that for the $\{\hat{R}_{ij}\}$ in Table~\ref{Correlation_Matrix_Graphon_c.elegans}. 
Indeed, the difference between the entries in Table~\ref{Correlation_Matrix_Graphon_c.elegans} and Table~\ref{Correlation_Matrix_SBM_c.elegans} are all less than $0.1$ and can be as small as $0.01$ or $0.03$. 
\begin{table}[tp]
    \centering
    \begin{tabular}{|l|c|c|c|}
        \hline
        & motor & inter & sensory\\
        \hline
  motor & 0.153 & 0.203 & 0.129\\
  inter & 0.203 & 0.145 &0.147\\
sensory & 0.129 & 0.147 & 0.171\\
        \hline
    \end{tabular}
    \caption{Pearson correlations between the edges in $A_c$
      and $A_g$ for different combinations of neuron types.}
    \label{Correlation_Matrix_SBM_c.elegans}
\end{table}

Finally, we discuss the use of the $\{\hat{R}_{ij}\}$ to help improve link predictions for the edges of $A_g$. More specifically, we evaluate the accuracy for link prediction using only the estimated edge probabilities matrix $\hat{P}$ against the accuracy when using $\hat{P}$ in conjunction with $\hat{R}$. For both approaches, we first sub-sample a $A_g^{(\mathrm{sub})}$ from $A_g$ by setting $10\%$ of the entries of $A_g$ to $0$. We next apply USVT to $A_g^{(\mathrm{sub})}$ to obtain an estimate $\hat{P}^{(\mathrm{sub})}$ of $P$. Now let $\mathcal{E}$ be the set of entries in $A_g$ that are set to $0$ in $A_g^{\mathrm{sub}}$. We then threshold the entries $\hat{P}^{(\mathrm{sub})}_{ij}$ for all $(i,j) \in \mathcal{E}$, i.e., for $(i,j) \in \mathcal{E}$ we predict the presence of a link if $\hat{P}^{(\mathrm{sub})}_{ij} > t$ and an absence of a link otherwise. By varying $t \in [0,1]$ we obtain a ROC curve and an
associated AUC for link prediction using only the estimated $\hat{P}^{(\mathrm{sub})}$. A similar approach had also been used in
\cite{yuan_zhang,gao2015rate,rubin2017statistical} when the graphs are assumed to be generated from a latent space or graphon model. Link prediction using both $A_c$ and $A_g$ also follows a similar procedure, but this time we use both $A_g^{(\mathrm{sub})}$ and $A_c^{(\mathrm{sub})}$ to estimate the marginal edge probabilities $\hat{P}^{(\mathrm{sub})}$ for $A_g$ and $\hat{Q}^{(\mathrm{sub})}$ for $A_c$ as well as the estimated correlations $\hat{R}^{(\mathrm{sub})}$; here $A_c^{(\mathrm{sub})}$ is obtained by setting the entries of $A_c$ indexed by $\mathcal{E}$ to $0$ and $\hat{R}^{(\mathrm{sub})}$ is calculated using a similar expression as that in Eq.~\eqref{eq:estimate_R} but with $\hat{P}$ and $\hat{Q}$ replaced by $\hat{P}^{(\mathrm{sub})}$ and $\hat{Q}^{(\mathrm{sub})}$, respectively. Given the $\hat{P}^{(\mathrm{sub})}$ and $\hat{R}^{(\mathrm{sub})}$ we then threshold the entries of $\hat{P}^{(\mathrm{sub})}_{ij} + \hat{R}_{ij}^{(\mathrm{sub})}( A_c(i,j)  - \hat{P}_{ij}^{(\mathrm{sub})})$ for all $(i, j) \in \mathcal{E}$; note that these choices of quantities is motivated from the fact that if $(A, B) \sim R$-$ER(P,Q)$ then $\mathbb{P}[A_{ij} = 1 \mid B_{ij}] = P_{ij} + R_{ij}(B_{ij} - P_{ij})$. By varying the threshold $t \in [0,1]$ we also obtain a ROC curve and associated AUC for link prediction using both $\hat{P}^{(\mathrm{sub})}$ and $\hat{R}^{(\mathrm{sub})}$.

We perform the above AUC calculations $100$ times, each time choosing a random subset of entries $\mathcal{E}$ to set to $0$. ROC curves for a random realization of $\mathcal{E}$ are shown in Figure \ref{C.elegans_roc}. The average AUC when using only $\hat{P}^{(\mathrm{sub})}$ is $0.575$ with a standard error of $0.002$; in contrast, the average AUC when using both $\hat{P}^{(\mathrm{sub})}$ and $\hat{R}^{(\mathrm{sub})}$ is $0.705$
with a standard error of $0.003$. The use of $\hat{R}^{(\mathrm{sub})}$ thus leads to a significant increase in accuracy. 
\begin{figure}[h]
    \centering
    \includegraphics[width=0.8\linewidth]{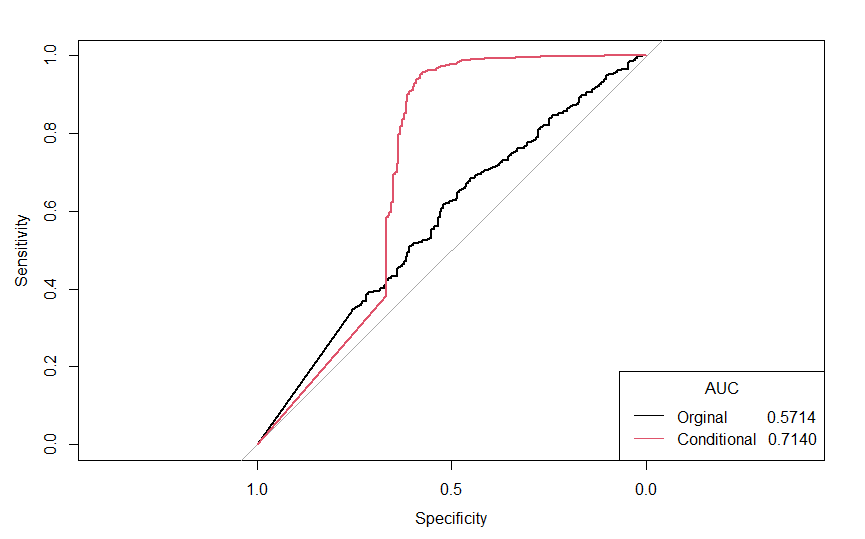}
    \caption{ROC curves for link prediction on a randomly selected set
      of entries $\mathcal{E}$ of the {\em C. elegans} gap junction network $A_g$. The black curve is for
      using $A_g$ only while the red curve is for using both $A_g$ and $A_c$. }
    \label{C.elegans_roc}
\end{figure}

\begin{algorithm}[tp]  
  \caption{Bootstrap procedure for graphons with possibly $P \not = Q$.}  
  \label{Analysis_for_real_data}  
  \begin{algorithmic}
    \Require Adjacency matrices $A$ and $B$, both of size $n \times n$, significance level $\alpha \in (0,1)$, number of bootstrap samples $m$.
     \State (A) Compute the matrix $C$ whose entries are $C_{ij} = 1$ if $A_{ij} + B_{ij} > 0$ and $C_{ij} = 0$ otherwise. 
    \State (B) Compute $\hat{P}, \hat{Q}$ and $\hat{H}$ by applying universal singular value thresholding (USVT) on $A$, $B$, and $C$, respectively. 
    \State (C) Calculate the test statistic $T = \|\hat{H} -  \hat{P} -\hat{Q} + \hat{P} \circ \hat{Q}\|_{F}$.
    \For{$s=1$ to $m$}
       \State (i) Generate adjacency matrices $(A^{(s)}, B^{(s)})$ according to Definition 5 with $R = 0$ and marginal edge probabilities matrices $\hat{P}$ and $\hat{Q}$. 
        \State (ii) Compute $\hat{P}$ and $\hat{P}$ as the universal singular value threshold of $A^{(s)}$ and $B^{(s)}$, respectively. 
       \State (iii) Calculate $T^{(s)} = \|\hat{H}^{(s)}-  \hat{P}^{(s)}-\hat{Q}^{(s)} + \hat{P}^{(s)} \circ \hat{Q}^{(s)}\|_{F}$, where $\hat{H}^{(s)}$ is the universal singular value thresholding of $A^{(s)} + B^{(s)}-A^{(s)} \circ B^{(s)}$.
      \EndFor
      \State (D) Find the smallest number $t$ such that $T>T_t$, where $T_t$ is the $t$-th largest element in $\{T^{(s)}\}_{s=1}^{m}$,
    \State (D) p-value$=(t-0.5)/m$
    
     \State \textbf{Output} p-value. 
  \end{algorithmic}  
\end{algorithm}

\subsection{Wikipedia Data}
We now analyze two networks formed by a collection of Wikipedia articles. The first network, denote as $A_e$, consists of $1382$ vertices and $37714$ edges. Each vertex in $A_e$ represents an article in the English Wikipedia on topics related to Algebraic Geometry, and two given vertices are connected if there is a hyperlink between them in the English Wikipedia. The second network, denote as $A_f$, consists of $1382$ vertices and $29946$ edges corresponding to the same Wikipedia articles as that in $A_e$ but the hyperlinks are now for the French Wikipedia. See \cite{priebe2009fusion} for a more detailed description of these networks. We now follow the same analysis as that described in Section~\ref{Analysis_for_real_data} for the {\em C. elegans data}. In
particular, we first test the null hypothesis that $A_e$ and $A_f$ are independent. As $A_e$ and $A_f$ are both quite sparse (their edge densities are $0.02$ and $0.016$, respectively), we apply the test statistic in Theorem~\ref{theorem 1} to their complements $\bar{A}_e$ and $\bar{A}_f$, i.e., $\bar{A}_e = 11^{\top} - A_e$ and
$\bar{A}_f = 11^{\top} - A_f$ where $11^{\top}$ is the $1382 \times 1382$ matrix of all ones. Note that, by Eq.~\eqref{distribution P_n}, if $(A, B) \sim R$-$ER(P)$ then $(\bar{A}, \bar{B}) \sim R-ER(11^{\top} - P)$ and hence, assuming the model in Section~\ref{graphon Model} is appropriate, inference based on $T(A, B)$ and $T(\bar{A}, \bar{B})$ are theoretically equivalent. This yield an observed test statistic of $T(\bar{A}_e, \bar{A}_f) = 21.514$ and, using the bootstrapping procedure in Algorithm~\ref{Analysis_for_real_data} with $m = 10000$, an approximate $p$-value of $5 \times 10^{-5}$. We thus reject the null hypothesis and are in favor of the alternative hypothesis that the English and French Wikipedia networks are correlated. 

We next quantify the degree of correlations between the edges of $A_e$ and $A_f$. The articles in $A_e$ and $A_f$
can be grouped into six classes, namely (1) people, (2) places, (3) dates, (4) math things (articles about math topics that are neither people, places, nor dates) (5) things (article about non-math topics that are neither people, places nor dates) and (6) categories (a special type of Wikipedia article). We then calculate the sample means of the estimated correlations $\hat{R}$ for edges within the same categories and between different categories (see the description on Table~\ref{Correlation_Matrix_Graphon_c.elegans}
 and Table~\ref{Correlation_Matrix_SBM_c.elegans} in Section~\ref{sec:c_elegans} for more details). The results are presented in Table~\ref{Correlation_Matrix_Graphon_Wiki} and
Table~\ref{Correlation_Matrix_SBM_Wiki}; once again we see that the entries in the two tables are highly similar and they both indicate that the correlations between the edges of $A_e$ and $A_f$ are positive and quite large. 
\begin{table}[h]
    \centering
    \begin{tabular}{|l|c|c|c|c|c|c|}
        \hline
        & people & places & dates & things & math things & categories\\
        \hline
    people & .501 & .419 &.336 &.381 &.375 &.417\\
    places & .419 & .318 &.269 &.296 &.272 &.307\\
     dates & .336 & .269 &.278 &.227 &.148 &.159\\
    things & .381 & .296 &.227 &.267 &.238 &.282\\
math things& .375 & .272 &.148 &.238 &.192 &.231\\
categories & .417 & .307 &.159 &.282 &.231 &.273\\
        \hline
    \end{tabular}
    \caption{Sample means of the estimated correlations
      $\{\hat{R}_{ij}\}$ for different combinations of Wikipedia
      article types for $i$ and $j$.}
    \label{Correlation_Matrix_Graphon_Wiki}
\end{table}

\begin{table}[h]
    \centering
    \begin{tabular}{|l|c|c|c|c|c|c|}
        \hline
        & people & places & dates & things & math things & categories\\
        \hline
    people & .570 & .501  &.460  &.399 &.499 &.751\\
    places &.501 & .489 & .376 & .412 & .295 & .607\\
     dates & .460 & .376 & .589 & .292 & .100 &    NA\\
    things & .399 & .412 & .292 &.348  &.233 & .452\\
math things& .499 & .295 & .100 &.233 & .283 & .507\\
categories & .751 & .607  &  NA & .452 &.507 &.462\\
        \hline
    \end{tabular}
    \caption{Pearson correlations between the edges of $A_e$ and $A_f$
	for different combinations of Wikipedia article types. The value \texttt{NA} for the pair \texttt{dates} and \texttt{categories} is because there are no edges between any vertices in \texttt{date} and any 
	vertices in \texttt{categories} for both graphs.}
    \label{Correlation_Matrix_SBM_Wiki}
\end{table}

Finally, we consider link prediction for the English Wikipedia network $A_e$. We follow the procedure described in Section~\ref{Analysis_for_real_data} wherein we set $10\%$ of
the entries of $A_e$ to $0$ and then compare the AUC for link prediction from $\hat{P}^{(\mathrm{sub})}$ alone against that of $\hat{P}^{(\mathrm{sub})}$ and $\hat{R}^{\mathrm{sub}}$. The sample mean of the AUCs based on $100$ randomly selected $\mathcal{E}$ are $0.863$ (standard error = $0.0006$) for $\hat{P}^{(\mathrm{sub})}$ only
and improve to $0.956$ (standard error = $0.0005$) when we also include $\hat{R}^{(\mathrm{sub})}$. ROC curves for a random realization of $\mathcal{E}$ are shown in Figure \ref{wiki_roc}.

\begin{figure}[h]
    \centering
    \includegraphics[width=0.8\linewidth]{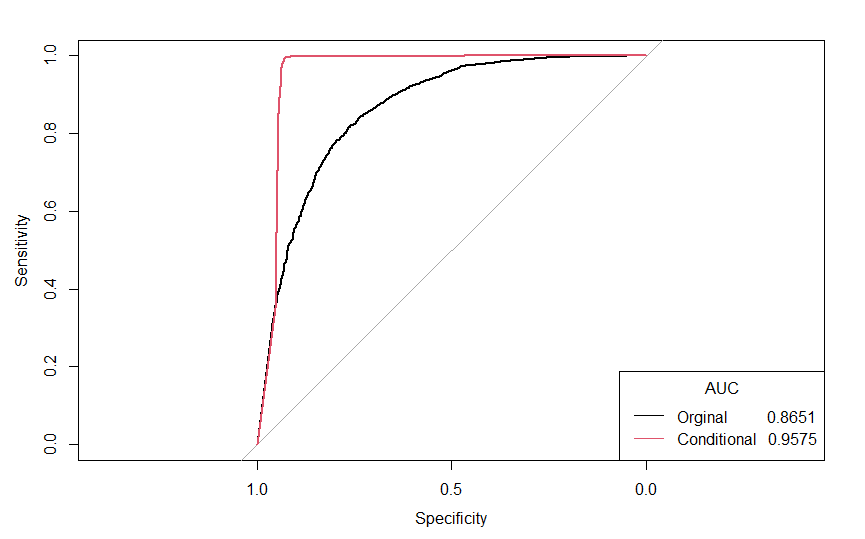}
    \caption{
ROC curves for link prediction on a randomly selected set
      of entries $\mathcal{E}$ of the English Wikpedia network $A_e$. The black curve is for
      using $A_e$ only while the red curve is for using both $A_e$ and $A_f$. 
}
    \label{wiki_roc}
\end{figure}

\section{Conclusion}
\label{sec:conclusion}
In this paper, we formulated independence testing between graphs as, given a pair of inhomogeneous Erd\H{o}s-R\'{e}nyi graphs with edge-correlation $R$, deciding between $\mathbb{H}_0 \colon \|R\|_{F} = 0$ and $\mathbb{H}_A \colon \|R\|_{F} > 0$. We show that there exists an asymptotically valid and consistent test procedure only if $\|R\|_{F}\to\infty$ as
the number of vertices $n$ diverges. When the graphs and their pairwise correlations are generated from a latent position model, we propose an asymptotically valid and consistent test procedure that also runs in time polynomial in $n$. We now mention two directions for future research.

Comparing the theoretical results in Theorem~\ref{theorem 1} and Theorem~\ref{thm:SBM} against either Remark~\ref{rem:not_sufficient} or Theorem~\ref{thm:gap}, we see that while $\|R\|_{F} \rightarrow \infty$ for all of these examples, it is nevertheless easier, both statistically and computationally, to detect $R \not = 0$ when it has some structure. Indeed, for both Theorem~\ref{theorem 1} and Theorem~\ref{thm:SBM} we have $R_{ij} = f(m_{ij})$ where $f$ is a smooth function and the matrix $M = (m_{ij})$ is low-rank. In contrast, the matrix $R$ in Remark~\ref{rem:not_sufficient} and Theorem~\ref{thm:gap} is either completely random or has no low-rank structure.
Therefore, while $R \not = 0$ if and only if $\|R\|_{F} > 0$, the magnitude of $\|R\|_{F}$ itself is not sufficiently refined to distinguish between the simple and more difficult settings for $R \not = 0$. Determining the right measure of the correlation between graphs is thus of both theoretical and practical interest, especially if this measure also leads to thresholds that are both necessary and sufficient for our independence testing problem.

Continuing on the above theme, the critical region for our test statistics in Section~\ref{general} and Section~\ref{different} are based on bootstrapping graphs from the estimated edge probabilities matrices (see e.g., Algorithm~
\ref{Graphon_Model_Testing_Simulation_2}). The validity of these resampling techniques is justified by the empirical simulation studies as well as real data analysis. However, bootstrap sampling of a graph on $n$ vertices generally requires $O(n^2)$ time and $O(n^2)$ memory, which can be prohibitive if $n$ is large. Therefore our test procedures could be more robust and computationally efficient if we are able to derive the limiting distribution of the test statistics in Theorem~\ref{theorem 1} and Theorem~\ref{theorem 4} and thereby obtain approximate critical values. We surmise, however, that this will be a quite technical and challenging problem as it requires substantial refinement of all existing results for USVT as these exclusively focus on upper bounds for the estimation error in Frobenius norm.

Finally, it will also be useful to study other formulations of independence testing for graphs, e.g., by not assuming that they are marginally inhomogeneous Erd\H{o}s-R\'{e}nyi graphs, or by considering more complex correlation structures. A natural and interesting example of this latter type of problem is when we have three or more graphs as their joint distributions cannot be specified using only the marginal distributions and pairwise edges correlations.

\bibliographystyle{apalike}
\bibliography{biblio}

\newpage
\appendix

\section*{Proofs of Stated Results}
\label{sec:proofs_stated_results}
{\bf Proof of Theorem~\ref{thm:main1}.}
Let $\mathcal{S} = \{(0,0),(0,1),(1,0),(1,1)\}$. Let $t =
n(n-1)$. Now denote by
$$\mathcal{S}^{t} = \{(s_1, s_2, \dots, s_t) \colon s_i \in \mathcal{S}\}$$
the set of tuples of
length $t$ whose elements are from $\mathcal{S}$. We can then view any 
realization $(A_{n}, B_n)$ from the $R$-ER$(P)$ model as 
corresponding to some element of $\mathcal{S}^{t}$. Let $A_{ij}$
and $B_{ij}$ denote the $ij$th element of $A_{n}$ and $B_{n}$,
respectively; note that, for ease of exposition, we dropped the
index $n$ from these notations. The second moment for the likelihood
ratio between $\mathcal{P}_n$ and $\mathcal{Q}_n$ is then given by
\begin{equation*}
  \begin{split}
    \mathbb{E}_{\mathcal{Q}_n}\Big[\Big(\frac{\mathcal{P}_n(A_n, B_n)}{\mathcal{Q}_n(A_n,B_n)}\Big)^2\Big]
        &=\sum_{
             (A_{n},B_{n}) \in \mathcal{S}^{t}} 
                 \frac{\mathbb{P}(A_n,B_n)^2}{\mathbb{Q}(A_n,B_n)} \\
&=\sum_{(A_{n},B_{n}) \in \mathcal{S}^{t}} 
          \prod_{i< j}\frac{{\mathbb{P}(A_{ij},B_{ij})^2}}{\mathbb{Q}(A_{ij}, B_{ij})}\\
        &= \prod_{i < j}\Big[\Big(\frac{{\mathbb{P}}_{ij}(1,1)^2}{{\mathbb{Q}}_{ij}(1,1)}+
        \frac{{2\mathbb{P}}_{ij}(1,0)^2}{{\mathbb{Q}}_{ij}(1,0)}+\frac{{\mathbb{P}}_{ij}(0,0)^2}{{\mathbb{Q}}_{ij}(0,0)}\Big) \\
        &=\prod_{i < j}(1+R_{ij}^2)
\end{split}
\end{equation*}
The last equality in the above display is derived as follows (see also Eq.~\eqref{distribution P_n})
    \begin{equation*}
      \begin{split}
        \frac{\mathbb{P}_{ij}(1,1)^2}{\mathbb{Q}_{ij}(1,1)}
            + \frac{2\mathbb{P}_{ij}(1,0)^2}{\mathbb{Q}_{ij}(1,0)}
            + \frac{\mathbb{P}_{ij}(0,0)^2}{\mathbb{Q}_{ij}(0,0)} &= P_{ij}^2
            +2 P_{ij}(1-P_{ij})R_{ij}
            +(1-P_{ij})^2R_{ij}^2\\
        &+2P_{ij}(1-P_{ij})
            - 4 P_{ij}(1-P_{ij})R_{ij}
            \\ &+2P_{ij}(1-P_{ij})R_{ij}^2 \\
        &+(1-P_{ij})^2
            +2P_{ij}(1-P_{ij})R_{ij}
            +P_{ij}^2R_{ij}^2\\
        &=1+R_{ij}^2
    \end{split}
    \end{equation*}
We thus obtain
\begin{equation}
  \label{eq:2nd_moment}
  \begin{split}
    \mathbb{E}_{\mathcal{Q}_n}\Big[\Big(\frac{\mathcal{P}_n(A_n,B_n)}{\mathcal{Q}_n(A_n,B_n)}\Big)^2\Big]
    =\prod_{i < j}(1+R_{ij}^2).
  \end{split}
\end{equation}
Theorem~\ref{thm:main1} follows directly from Eq.~\eqref{eq:2nd_moment}
and the following technical lemma.
\begin{mylem}
  \label{lem:technical}
  $\limsup \prod_{i < j} (1 + R_{ij}^2) < \infty$ iff $\limsup
  \|R\|_F^2 < \infty$. 
\end{mylem}
\begin{proof}
First suppose that $\limsup_{n \rightarrow \infty}
\|R\|_F^2\leq C$ for some finite constant $C>0$. Denote $N = n(n-1)/2$. 
Then by Jensen’s inequality we have, for all 
but a finite number of $n$, that 
\begin{equation*}
  \begin{split}
\log \Bigl(\prod_{i < j}(1+R_{ij}^2)\Bigr) &=
    \sum_{i < j}\log(1+R_{ij}^2) \\ &\leq
                                         N\log\Bigl(1+\frac{1}{N}\sum_{i
                                         < j}R_{ij}^2\Bigr)\\
    &\leq \log\Big[\Bigl(1+\frac{1}{2N}\|R\|_F^2\Bigr)^{N}\Big]
    \leq\log\Big[\Bigl(1+\frac{C}{2N}\Bigr)^{N}\Big]
    \leq C/2
\end{split}
\end{equation*}
Conversely, suppose $\limsup_{n \rightarrow \infty}
\sum_{i < j}\log(1+R_{ij}^2)\leq C$ for some finite constant $C>0$.
Then, as $R_{ij}^2 \leq 1$ for all $\{i,j\}$ and $R_{ii} = 0$ for all $i$, we have
\begin{equation*}
  \begin{split}
    \frac{1}{2}\|R\|_F^2
    & \leq \sum_{i < j}\Bigl[e^{\log(1+R_{ij}^2)}-1\Bigr]\\
    &=\sum_{k=1}^\infty\sum_{i < j}\frac{\log^k(1+R_{ij}^2)}{k!}
    \\ & \leq \sum_{k=1}^\infty\sum_{i < j}\frac{\log(1+R_{ij}^2) \times \log^{k-1} 2}{k!} 
    \\ & =\Bigl(\sum_{i < j}\log(1+R_{ij}^2)\Bigr) \sum_{k=1}^\infty\frac{\log^{k-1} 2}{k!}
    \leq \frac{C}{\log 2}.
\end{split}
\end{equation*}
as desired. \end{proof}

	

{\bf Proof of Theorem~\ref{thm:gap}.}
Suppose we are given an adjacency matrix $S$ sampled from a planted
clique model with edges probability $p =
\tfrac{1}{2}$ and {\em unknown} clique size $s_0$. Let us generate a pair of {\em undirected} random graphs $(A, B)$
as follows. The collection $\{(A_{ij}, B_{ij})\}$ for $i < j$ are {\em
  independent} bivariate random variables and furthermore, for any
pair $i < j$,
\begin{gather*}
 \mathbb{P}(A_{ij} = B_{ij} = 1 \mid S_{ij} = 0) =
 \mathbb{P}(A_{ij} = B_{ij} = 0 \mid S_{ij} = 0) = 0.5,\\
 \mathbb{P}(A_{ij} = 1, B_{ij} = 0 \mid S_{ij} = 1) = \mathbb{P}(A_{ij} = 0, B_{ij} =
 1 \mid S_{ij} = 1) = 0.5.
 \end{gather*}
 Let $\xi_{ij} = 1$ if vertices $i$ and $j$ is part of the planted
 clique in $S$ and $\xi_{ij} = 0$ otherwise. Note that we can view the $\xi_{ij}$ as deterministic quantities by
 assuming that the vertices forming the planted
 clique are chosen prior to adding the random edges in $S$. In particular, 
 $S_{ij} = 1$ whenever $\xi_{ij} = 1$ and $S_{ij} \sim
 \mathrm{Bernoulli}(0.5)$ otherwise. We therefore have
 \begin{equation*}
   \begin{split}
   \mathbb{P}(A_{ij} = 1, B_{ij} = 1) &= \mathbb{P}(A_{ij} = B_{ij} = 1
   \mid S_{ij} = 0) \times \mathbb{P}(S_{ij} = 0) \\ &= 0.5 \times
   \bigl(\mathbb{P}(S_{ij} = 0, \xi_{ij} = 0) + \mathbb{P}(S_{ij} = 0,
   \xi_{ij} = 1)\bigr) \\ &= \tfrac{1}{4} \times \bm{1}\{\xi_{ij} = 0\}.
   \end{split}
 \end{equation*}
 Similar reasonings yield
 \begin{gather*}
   \mathbb{P}(A_{ij} = 0, B_{ij} = 0) = \mathbb{P}(A_{ij} = B_{ij} = 0
   \mid S_{ij} = 0) \times \mathbb{P}(S_{ij} = 0) = \tfrac{1}{4} \times
   \bm{1}\{\xi_{ij} = 0\}, \\
   \mathbb{P}(A_{ij} = 0, B_{ij} = 1) = 0.5 \times \mathbb{P}(S_{ij} =
   1) = \tfrac{1}{4} \times \bm{1}\{\xi_{ij} = 0\} + \tfrac{1}{2} \times \bm{1}\{\xi_{ij} = 1\}, \\
   \mathbb{P}(A_{ij} = 1, B_{ij} = 0) = \tfrac{1}{4} \times
   \bm{1}\{\xi_{ij} = 0\} + \tfrac{1}{2} \times \bm{1} \{\xi_{ij} = 1\},
 \end{gather*}
and hence $(A, B)$ is a realization of a $R$-correlated Erd\H{o}s-R\'{e}nyi graph
with edges probability $p = \tfrac{1}{2}$ and $R_{ij} = -\xi_{ij}$ for
all $\{i,j\}$
(see Eq.~\eqref{distribution P_n}). 
        
Now suppose that either $s_0 = 0$ or
$s_0 = n^{1/4}$. Then, given $S$ and the pair $(A, B)$ randomly generated from $S$, we have 
       \begin{align*}
	        	\mathbb{H}_0^{(1)} \colon \|R\|_F = 0	& \Longleftrightarrow	
                                                                       \mathbb{H}_0^{(2)} \colon
                                                                       \text{$S$ has no planted clique}\\
	        	\mathbb{H}_A^{(1)} \colon\|R\|_F = n^{1/4}	&
                                                    \Longleftrightarrow 
                                                                       \mathbb{H}_A^{(2)}
                                                                       \colon
                                                                       \text{$S$
                                                                       has
                                                                       a
                                                                       planted
                                                                       clique
                                                                       of
                                                                       size
                                                                       at
                                                                       least
                                                                       $n^{1/4}$.}
       \end{align*}
        Therefore, for any given instance $S \sim \text{PlantedClique}(n,
        1/2, s_0)$, there exists an instance $(A, B)$ from
        $R$-$\mathrm{ER}(1/2)$ where $R$ is such that (1) the pair $(A,
        B)$ is generated in polynomial time and (2)
        the planted clique problem on $S$ is equivalent to deciding
        between the null and alternative hypothesis for $\|R\|_F$. Thus, assuming
        the Planted Clique conjecture holds, i.e., 
        the PlantedClique problem requires quasi-polynomial time, there is no efficient algorithm
        for deciding between $\|R\|_F = 0$ versus $\|R\|_F > 0$ in
       this setting. 

       {\bf Proof of Theorem \ref{theorem 1}}
Let $h_n = \rho_n h$ and $g_n = \gamma_n g$. Now recall
Eq.~\eqref{eq:graphon_H}. Then $C$ corresponds to the adjacency matrix 
of a latent position graph with latent positions $\{(X_i,
Y_i)\}_{i=1}^{n}$ and link function $f_n \colon \mathbb{R}^{d + d'}
\times \mathbb{R}^{d + d'} \mapsto [0,1]$ given by
\begin{equation}\label{eq:graphon_Ha}
 f_n((x,y),(x',y')) = h_n(x,x') + (1 - g_n(y,y')) h_n(x,x')(1 - h_n(x,x')).
\end{equation}
      Next suppose that $\|R\|_{F} = 0$. Then by Theorem~3 in \cite{xu2017rates}, for all $c>0$ there
             exists a constant $C$ such that with probability at least
             $1 - n^{-c}$
    \begin{eqnarray}
    \label{eq:theorem1_xu1}
    	\| \hat{P} - P \|_F \leq C (n \rho_n)^{\tfrac{s + d}{2s + d}}, \quad \text{and} \quad \| \hat{H} - 2P + P \circ P \|_F
      \leq C (n \rho_n)^{\tfrac{s+d}{2s+d}}.
    \end{eqnarray}
    simultaneously. Similarly, suppose $\|R\| > 0$ holds. Once again by Theorem~3 in
    \citet{xu2017rates}, there exists a constant $C'$ such that
    with probability at least $1 - n^{-c}$,
    \begin{eqnarray}
    \label{eq:theorem1_xu2}
    	\| \hat{P} - P \|_F \leq C (n \rho_n)^{\tfrac{s + d}{2s + d}}, \quad \text{and} \quad \|
      \hat{H} - H \|_F \leq
      C (n \rho_n)^{\tfrac{s+d+d'}{2s + d + d'}}
    \end{eqnarray}
    simultaneously. We note that the upper bound for
    $\|\hat{H} - H\|_{F}$ in Eq.~\eqref{eq:theorem1_xu2} is larger
    than that in Eq.~\eqref{eq:theorem1_xu1} and this is due mainly to
    the fact that if $\|R\|_{F} > 0$ then we will be using the latent positions $Z_i = (X_i, Y_i) \in
    \mathbb{R}^{d+d'}$ together with a link function $f_n$ that is also at least
    $s$ times continuously differentiable.
    We therefore have, with probability at least $1 -
    n^{-c}$, that
	\begin{eqnarray*}
		\| \hat{P}\circ\hat{P} - P\circ P \|_F^2
		=	\sum_{i, j}	(\hat{P}_{ij} + P_{ij})^2(\hat{P}_{ij} - P_{ij})^2
		\leq	4\| \hat{P} - P \|_F^2
		\leq	4C^2 (n\rho_n)^{\tfrac{2s+2d}{2s + d}}.
	\end{eqnarray*}
	and hence, for $\|R\| = 0$ we have
	\begin{eqnarray*}
		\|\hat{H} - 2 \hat{P} + \hat{P}\circ\hat{P}\|_F
		&\leq&	\|\hat{H} - 2P + P \circ P \|_F + 2\|\hat{P} - P\|_F +
      \|\hat{P} \circ \hat{P} - P \circ P\|_{F} \\
		&\leq& 4C(n\rho_n)^{\tfrac{s+d}{2s+d}}
	\end{eqnarray*}
    with probability at least $1 - n^{-c}$. Similarly, for $\|R\| > 0$
    we have
	\begin{equation*}
 \begin{split}
		\|\hat{H} - 2 \hat{P} + \hat{P} \circ \hat{P}\|_F
		&\geq \|H - 2P - P\circ P\|_F - \bigl(\|\hat{H} - H\|_F + 2
               \|\hat{P} - P\|_{F} + \|\hat{P}\circ\hat{P} - P\circ P\|_F\bigr) \\
            &\geq \|H - 2P - P\circ P\|_F  - 4C(n \rho_n)^{\tfrac{s+d+d'}{2s+d+d'}} \\
            &= \|R \circ (P - P \circ P)\|_{F} - 4C(n \rho_n)^{\tfrac{s+d+d'}{2s+d+d'}}
	\end{split}
 \end{equation*}
    Now let $\alpha = \tfrac{s + d + d'}{2s + d + d'}$ and note that $\frac{s+d}{2s+d} \leq \alpha$ for any choice of $s \geq 0, d \geq 0$ and $d' \geq 0$.
    Define $T(A,B)$ as the test statistic
    $$T(A, B) = \frac{\|\hat{H} - 2 \hat{P} + \hat{P} \circ
    \hat{P}\|_{F}}{(n \rho_n)^{\alpha} \log^{1/2}{n}}.$$
    If $\|R\|_{F} = 0$ then $T(A, B) \rightarrow 0$ as $n
  \rightarrow \infty$. Furthermore
  if $\|R \circ (P - P \circ P)\|_{F} = \Omega((n \rho_n)^{\alpha'})$ for any
    $\alpha' > \alpha$ then $T(A, B) \rightarrow \infty$ as $n \rightarrow
    \infty$. Thus rejecting $\mathbb{H}_0$ for large values of $T(A,
    B)$ leads to an asymptotically valid and consistent test. 

    {\bf Proof of Corollary~\ref{cor:lpg}}
            If $f$ and $g$ are infinitely differentiable then, in place of Eq.~\eqref{eq:theorem1_xu1} and Eq.~\eqref{eq:theorem1_xu2}, we have
             $$\|\hat{P} - P\|_{F} = O((n \rho_n)^{1/2} \log^{d/2}(n \rho_n)), \quad \|\hat{H} - 2P + P \circ P\|_{F}$$
             under $\mathbb{H}_0$ and
             $$\|\hat{P} - P\|_{F} = O((n \rho_n)^{1/2} \log^{d/2}(n \rho_n)), \quad \|\hat{H} - H\|_{F}$$
             under $\mathbb{H}_A$. See Theorem~4 in \cite{xu2017rates} for a statement of these bounds. The remaining steps follow the same argument as that presented in the proof of Theorem~\ref{theorem 1}. We omit the details. 

{\bf Proof of Theorem~\ref{thm:SBM}}
Recall that the vertices of $C$ are clustered using a community detection algorithm which guarantees exact recovery (see e.g., \cite{abbe2017community,gao2017achieving,lyzinski2014perfect}). We therefore have $\hat{\tau} = \tau$ asymptotically almost surely. Let us now condition on the event that $\hat{\tau} = \tau$. Then for any $k, \ell \in \{1,2,\dots,K\}$, the collection $\{(A_{ij}, B_{ij}) \colon \tau_i = k, \tau_{j} = \ell\}$ are iid bivariate random vectors with Pearson correlation $\rho_{k \ell}$. The central limit theorem then implies
$$\sqrt{n_{k \ell}} (\hat{\rho}_{k\ell}-{\rho}_{k\ell}) \overset{d}{\to}\mathcal{N}(0,(1 - \rho_{k \ell}^2)^2).$$
Furthermore, as $\hat{\rho}_{k \ell}$ depends only on the edges from vertices in the $k$th block to vertices in the $\ell$th block, the $\{\hat{\rho}_{k \ell}\}$ are {\em mutually} independent. 
Now suppose that the null hypothesis is true. Then $\rho_{k \ell} \equiv 0$ and 
the $\sqrt{n_{k \ell}} \hat{\rho}_{k \ell}$ are iid standard normals. In other words we have
$$\sum_{k \leq \ell} n_{k \ell} \hat{\rho}^{2}_{k \ell} \overset{d}{\rightarrow} \chi^2_{K(K+1)/2}$$
under $\mathbb{H}_0$. Next suppose that the alternative hypothesis is true and that there exists a constant $\mu > 0$ such that 
$$\sum_{k \leq \ell} n_{k \ell} \rho_{k \ell}^2 \rightarrow \mu.$$
Then for any $k, \ell$, the term $\sqrt{n_{k \ell}} \rho_{k \ell}$ is bounded and thus, by Slutsky's theorem, we have
$$\sqrt{n_{k \ell}}(\hat{\rho}_{k \ell} - \rho_{k \ell}) \overset{d}{\rightarrow} \mathcal{N}(0,1)$$
for all $k, \ell$.
We can then follow the same arguments as that in \cite{guenther1964another} and show that the limiting distribution of $\sum_{k \leq \ell} n_{k \ell} \hat{\rho}_{k \ell}^2$ depends on the $\{\rho_{k \ell}\}$ only through the quantity $\sum_{k \leq \ell} n_{k \ell} \rho_{k \ell}^2$. Hence, as $\sum_{k \ell} n_{k \ell} \rho_{k \ell}^2 \rightarrow \mu$ we have by Slutksy's theorem that
%
$$\sum_{k\leq l}n_{kl}\hat{\rho}_{kl}^2 \overset{d}{\to} \chi^2_{K(K+1)/2}(\mu)$$
as desired.

\end{document}